\documentclass{article}

\nonstopmode

\usepackage{amsmath}
\usepackage{amssymb}

\usepackage{upgreek}
\usepackage[dvips]{graphicx}
\usepackage{amsthm}
\usepackage{mathrsfs}
\usepackage{bm}
\usepackage{color}
\usepackage[backref=none]{hyperref}
\usepackage{doi} 
\definecolor{MyDarkBlue}{rgb}{0,0.1,0.4}
\hypersetup{pdfborder={0 0 0},colorlinks,breaklinks=true,
  urlcolor={MyDarkBlue},linkcolor={MyDarkBlue},citecolor={MyDarkBlue}
}

\DeclareSymbolFont{EUR}{U}{eur}{m}{n}
\SetSymbolFont{EUR}{bold}{U}{eur}{b}{n}
\DeclareSymbolFontAlphabet{\eur}{EUR}

\DeclareSymbolFont{EUB}{U}{eur}{b}{n}
\SetSymbolFont{EUB}{bold}{U}{eur}{b}{n}
\DeclareSymbolFontAlphabet{\eub}{EUB}

\newcommand{\p}{\partial}

\newcommand\eubD{\eub{D}}
\newcommand\eubV{\eub{V}}
\newcommand\eubJ{\eub{J}}
\newcommand\eubL{\eub{L}}

\newcommand\bmupalpha{\bm\upalpha}
\newcommand\bmupbeta{\bm\upbeta}

\newcommand\bmupphi{\bm\upphi}
\newcommand\bmuprho{\bm\uprho}
\newcommand\bmupxi{\bm\upxi}
\newcommand\bmupzeta{\bm\upzeta}

\newcommand\R{\mathbb{R}}
\newcommand\N{\mathbb{N}}
\newcommand{\C}{\mathbb{C}}
\newcommand{\at}[1]{\vert\sb{\sb{#1}}}
\newcommand{\abs}[1]{\vert #1 \vert}

\newcommand{\norm}[1]{\left\|#1\right\|}

\renewcommand{\Re}{\mathop{\rm Re}}
\renewcommand{\Im}{\mathop{\rm Im}}

\theoremstyle{plain}
\newtheorem{lemma}{Lemma}

\theoremstyle{definition}

\theoremstyle{remark}
\newtheorem{remark}[lemma]{Remark}

\begin{document}


\title{
Vakhitov-Kolokolov and energy vanishing conditions
for linear instability of solitary waves in 
models of classical self-interacting spinor fields}

\author{
{Gregory Berkolaiko},$\!\sp{a}\,$
{Andrew Comech},$\!\sp{a,b}\,$
and
{Alim Sukhtayev}$\sp{a}$
\\
$\sp{a}${\small\it Texas A\&M University, College Station, TX 77843, U.S.A.}
\\
$\sp{b}${\small\it  Institute for Information Transmission Problems, Moscow 101447, Russia}
}



\date{\today}

\maketitle

\begin{abstract}
We study the linear stability of localized modes
in self-interacting spinor fields,
analyzing the spectrum of the operator
corresponding to linearization at solitary waves.
Following the generalization of the Vakhitov--Kolokolov approach,
we show that the bifurcation of real eigenvalues from the origin
is completely characterized
by the Vakhitov--Kolokolov condition $dQ/d\omega=0$
and by the vanishing of the energy functional.
We give the numerical data
on the linear stability
in the generalized Gross--Neveu model
and the generalized massive Thirring model
in the charge-subcritical, critical, and supercritical cases,
showing the agreement with the Vakhitov--Kolokolov
and the energy vanishing conditions.
\end{abstract}

\bigskip

\section{Introduction}
Models of self-interacting spinor fields
have been playing a prominent role in Physics
for a long time
\cite{jetp.8.260,PhysRev.83.326,PhysRev.103.1571,
RevModPhys.29.269}.
Widely considered are
the massive Thirring model (MTM) \cite{MR0091788},
the Soler model \cite{PhysRevD.1.2766},
and the massive Gross--Neveu model \cite{PhysRevD.10.3235,PhysRevD.12.3880},
in which
the self-interaction gives rise to
non-topological solitons of the form
$\phi\sb\omega(x)e^{-i\omega t}$.
Such solitary waves have also been found in
Dirac--Maxwell system (DM) \cite{wakano-1966,MR1364144,MR1386737,MR1618672}
and Dirac--Einstein systems 
\cite{PhysRevD.59.104020,MR2647868}.

We point out that
we treat the fermionic field classically, as a $c$-number,
completely leaving out the framework of the second quantization.
Because of this, we need to mention
the role played by such classical solitary wave solutions
in physics.
In
\cite{PhysRevD.10.4130,PhysRevD.12.3880}
such classical states were considered
from the point of view of
classical approximations of hadrons.
It was shown in \cite{MR3118823}
that 
the (classical) Dirac--Coulomb (DC) and DM systems
appear in the quantum field theory
where these solitary waves correspond to polarons,
formed due to interaction of fermions
with optical phonons or with the gravitational field \cite{MR766230}
(this reflects the Landau--Pekar approach
to the polaron
in the conventional nonrelativistic electron theory
\cite{polaron-1933,polaron-1946};
a similar mechanism is also responsible
for the formation of the Cooper pairs
in the microscopic mechanism of the superconductivity \cite{MR766230}).
Classical self-interacting
spinor field also appears in the Dirac--Hartree--Fock approach
in Quantum Chemistry
(see e.g. \cite{QUA:QUA20146}).
Coupled mode equations
in the nonlinear optics
and the theory of Bose-Einstein condensates
could also be treated as one-dimensional
Dirac equation with self-interaction
of a particular type
\cite{PhysRevLett.80.5117,MR2217129,MR2513792,PhysRevE.70.036618}.
The approach given in the present paper
is applicable to the stability analysis of
all such
systems of classical self-interacting spinors.

In spite of many attempts at stability
in the context of the classical spinor fields
(e.g.
\cite{MR592382,PhysRevLett.50.1230,MR848095,PhysRevD.34.644,PhysRevD.36.2422,2014arXiv1408.4171C}),
an exhaustive characterization of stability properties in these models
is still absent.
For many years, the stability analysis in the nonlinear
Schr\"odinger and other similar systems was based on
the Vakhitov--Kolokolov (VK) stability criterion \cite{VaKo},
with several important generalizations obtained
in \cite{MR783974,MR901236,MR1081647,MR2089513,MR2121936,MR2869071};
these results, however, were obtained for the systems
whose Hamiltonian is bounded below
and
do not extend to the Dirac-type systems,
only allowing some partial conclusions.
We point out that formally
(ignoring the unboundedness of the Hamiltonian in the spinor models)
our instability results
fit the general framework developed in a recent paper \cite{MR2869071}.
(Roughly, the conditions for the Jordan block type
degeneracies of the linearization at a solitary wave
remain the same, correctly describing the collisions of eigenvalues
at the origin, but one no longer knows which direction
the eigenvalues are located before and after the collision;
as a result, the count of the ``unstable''
positive-real-part
eigenvalues is no longer reliable.)

In view of recent results on
stability and instability for the nonlinear Dirac equation
\cite{PhysRevE.82.036604,MR2892774,linear-a,MR3208458}
it is becoming clear
that the VK criterion is still useful for the spinor systems in the
nonrelativistic limit,
when the amplitude of solitary waves is small.
In particular,
the ground states (``smallest energy solitary waves'')
in the charge-subcritical nonlinear Dirac equation
(with the nonlinearity of order $2k+1$, with $k<2/n$)
are linearly stable in the nonrelativistic limit
$\omega\lesssim m$,
which corresponds to solitary waves of small amplitudes.
The same linear stability
is expected to be true for the Dirac--Maxwell system
in the nonrelativistic limit $\omega\gtrsim -m$
\cite{dm-existence,MR3118823}.



In the present paper, following \cite{VaKo}
and \cite{MR1081647},
we show that
in the systems of self-interacting spinor fields
the condition
\begin{equation}
E(\phi\sb\omega)=0,
\end{equation}
alongside with the Vakhitov--Kolokolov condition
$dQ(\phi\sb\omega)/d\omega=0$,
indicates the collision of eigenvalues at the origin,
marking a possible border of the stability and instability regions.
Above, $E(\phi\sb\omega)$ and $Q(\phi\sb\omega)$
are the energy and the charge
of a corresponding solitary wave.

We then show that our theory applies
to the pure power generalized massive Thirring model,
with the nonlinearity of order $p=2k+1$, $k>0$.
Our numerical results show that
in all models with $k\ne 1$,
the energy functional
vanishes at some $\omega\sb k\in(-m,0)$
(with $\lim\sb{k\to 1}\omega\sb k=-m$).
We then compute the eigenvalues
of the linearizations at these solitary waves,
and show that there is a birth of a pair
of positive-negative eigenvalues
precisely at the value $\omega$
which corresponds to solitary waves
of zero energy.
On Figure~\ref{fig-mtm-0.5} below, we plot
the spectra of linearization
and the values of the energy for the solitary waves
in the generalized massive Thirring model with quadratic nonlinearity
($k=1/2$).

Let us mention that in the original, cubic massive Thirring model
(with $k=1$)
the solitary waves were recently shown to be orbitally stable in $H^1$
for $\omega\lesssim m$ \cite{2013arXiv1312.1019C},
and orbitally stable in $L^2$ for all $\omega\in(-m,m)$ \cite{2014LMaPh.104...21P}.
These results are based on the complete integrability of the (cubic) massive Thirring model.

We also mention the situation with the coupled-mode equations of the Dirac type which are
not Lorentz-invariant, such as in \cite{PhysRevLett.80.5117}.
The condition which describes
bifurcations of eigenvalues from the origin
is formulated in terms of vanishing of the determinant
consisting of the derivatives of the conserved
quantities at the solitary wave parameters;
in the case of coupled-mode equations,
the corresponding matrix is not diagonal
(see e.g. \cite{MR2869071}).
The present paper shows that
for the Lorentz-invariant nonlinear Dirac equations,
the corresponding matrix of derivatives is diagonal,
with the derivative of the momentum with respect to the speed of solitary waves
being proportional to the value of energy of the solitary wave.

Let us give an informal outline.
After the linearization at a solitary wave,
the isolated eigenvalue
$\lambda=0$ of the linearized equation
corresponds to several
Jordan blocks related to the symmetries of the system, most
importantly the $\mathbf{U}(1)$-invariance and the translational invariance.
When two purely imaginary eigenvalues collide at $0$, they do so by
joining one of these
Jordan blocks;
the collision then
produces a pair of real eigenvalues (one positive, one negative)
and results in linear (exponential) instability.
The VK condition $d Q/d\omega=0$ detects
the enlargement of the block corresponding to
the $\mathbf{U}(1)$ symmetry.
We study
the blocks corresponding to the translation invariance
and derive the condition for
their enlargement;
it turns out that the condition corresponds
to the energy of the solitary wave being zero.
Let us mention that similar methods of studying transition to instability
are employed in \cite{MR1081647,MR2089513}.

\bigskip

\section{Main results}
Let $\gamma\sp\mu$,
$0\le\mu\le n$,
be the $N\times N$ Dirac matrices
which satisfy $\{\gamma\sp\mu,\gamma\sp\nu\}=2 g\sp{\mu\nu}I_N$
($N\in\N$ is even),
with
$g\sp{\mu\nu}=\mathop{\rm diag}[1,-1,\dots,-1]$
the Minkowski metric.
For $\psi\in\C^N$,
denote $\bar\psi=\psi\sp\ast\gamma\sp 0$.
Consider the Lagrangian density
\begin{equation}\label{nld-lagrangian}
\mathscr{L}
=
\bar\psi(i\gamma\sp\mu\p\sb\mu-m)\psi
+\mathcal{F}(\bar\psi,\psi),
\end{equation}
with $m>0$, $\psi\in\C^N$ the spinor field,
and $\mathcal{F}:\;\C^N\times\C^N\to\C$,
which we assume is sufficiently smooth
and satisfies
$\abs{\mathcal{F}(\bar\psi,\psi)}=o(\abs{\psi}^2)$
for $\abs{\psi}\ll 1$.
We also assume that
$\mathcal{F}(\bar\psi,\psi)$
is $\mathbf{U}(1)$-invariant:
\[
\mathcal{F}(e^{-is}\bar\psi,e^{is}\psi)
=\mathcal{F}(\bar\psi,\psi),
\qquad
\psi\in\C^N,
\quad
s\in\R.
\]
The Euler-Lagrange equation
obtained by
taking the variation of \eqref{nld-lagrangian}
with respect to $\bar\psi$
(considered as independent of $\psi$)
leads to the equation
\begin{equation}\label{nld}
i\dot\psi
=D\sb m\psi-\beta\nabla\sb{\bar\psi}\mathcal{F}.
\end{equation}
Above,
$
D\sb m
=
-i\alpha\sp j\p\sb{j}+\beta m
$
is the Dirac operator,
with $\alpha\sp j=\gamma\sp 0\gamma\sp j$, $\beta=\gamma\sp 0$
the self-adjoint Dirac matrices.
We follow the convention that
$0\le \mu,\nu\le n$,
$1\le j,k\le n$,
and assume that there is a summation
with respect to repeated upper-lower indices
(unless specified otherwise).

\subsubsection*{Conservation laws and the Virial identity}

By N\"other's theorem,
due to the $\mathbf{U}(1)$-invariance
of the Hamiltonian,
there is a charge functional
\begin{equation}\label{eq:charge_def}
Q(\psi)=\int\psi\sp\ast(x,t)\psi(x,t)\,dx
\end{equation}
whose value is conserved along the trajectories.
(Here and below,
each integral is over $\R^n$ unless stated otherwise.)
The local law of charge conservation has the form
\begin{equation}\label{charge-local}
\p\sb\mu \mathscr{J}\sp\mu=0,
\end{equation}
where
\begin{equation}\label{def-j}
\mathscr{J}\sp\mu
=\bar\psi\gamma\sp\mu\psi
\end{equation}
is the four-vector of the charge-current density.

By \cite{MR0187642},
the density of the energy-momentum tensor
is given by
$
\mathscr{T}^{\mu\nu}
=
\frac{\p\mathscr{L}}
{\p(\p\sb\mu\psi)}
g^{\nu\rho}
\p\sb\rho\psi
-g^{\mu\nu}\mathscr{L}.
$
With
$
\mathscr{L}
=
\bar\psi(i\gamma\sp\mu\p\sb\mu-m)\psi
+\mathcal{F}(\bar\psi,\psi)
$
from \eqref{nld-lagrangian},
we have
\begin{equation}\label{nld-hamiltonian}
\mathscr{H}
=\mathscr{T}^{00}
=
-\bar\psi(i\gamma\sp j\p\sb j)\psi
+m\bar\psi\psi
-\mathcal{F}(\bar\psi,\psi).
\end{equation}
The components of the energy-momentum tensor
$T^{\mu\nu}=\int\mathscr{T}^{\mu\nu}\,dx$
are given by
\begin{eqnarray}\label{t-jk}
&
T^{00}=E=\int\mathscr{H}\,dx,
\nonumber
\\
&
T^{0k}=T^{k0}
=
i\int
\bar\psi
\gamma^0 g^{k\rho}\p\sb\rho\psi\,dx
=i g^{k\rho}\langle\psi,\p\sb\rho\psi\rangle,
\nonumber
\\
&
T^{jk}=
\big\langle\psi,i\alpha^{j}\p_\rho\psi\big\rangle
g^{\rho k}
-g\sp{jk}L,
\end{eqnarray}
where
$L=\int\mathscr{L}\,dx$.
Note that $T^{\mu\nu}$ is hermitian.

Now let us consider a solitary wave solution
\begin{equation}\label{sw}
\psi\sb\omega(x,t)=\phi\sb\omega(x)e^{-i\omega t},
\end{equation}
with $\phi\sb\omega(x)\in\C^N$ of Schwartz class in $x$.
Comparing
\eqref{nld-lagrangian}
and
\eqref{nld-hamiltonian},
we obtain:
\begin{equation}\label{l-h-q}
L(\psi\sb\omega)=-E(\psi\sb\omega)+\omega Q(\psi\sb\omega).
\end{equation}
By the Stokes theorem,
the local form of the charge conservation
\eqref{charge-local}
leads to
\[
0
=\p\sb t\textstyle\int \mathscr{J}\sp{0}x\sp k\,dx
=-\textstyle\int(\p\sb j\mathscr{J}\sp j)x\sp k\,dx
=
\textstyle\int\mathscr{J}\sp k\,dx.
\]
therefore,
for a solitary wave \eqref{sw}, one has:
\begin{equation}\label{j-zero}
{J}^{k}:=\int\mathscr{J}\sp k\,dx=0,
\qquad
1\le k\le n.
\end{equation}
Similarly,
since the energy-momentum tensor is the conserved N\"other
current associated with space-time translations,
for any fixed $0\le\nu\le n$,
there is the identity $\p\sb\mu \mathscr{T}^{\mu\nu}=0$
which follows from the Euler--Lagrange equations.
This leads to
\begin{equation}\label{t-zero}
T^{j\nu}=T^{\nu j}:=\int\mathscr{T}^{\nu j}\,dx=0,
\qquad
1\le j\le n.
\end{equation}


\medskip

We decompose the Hamiltonian functional into
\[
E(\psi)=K(\psi)+M(\psi)+V(\psi),
\]
with
\begin{eqnarray}\label{def-k-m-v}
&
K(\psi)=\int\psi\sp\ast(-i\alpha\cdot\nabla)\psi\,dx,
\qquad
M(\psi)=m\int\psi\sp\ast\beta\psi\,dx,
\qquad
V(\psi)
=
-\int\mathcal{F}(\bar\psi,\psi)\,dx.
\end{eqnarray}
Combining \eqref{t-jk} and \eqref{t-zero},
we conclude that
\begin{equation}\label{t-q}
i\int\phi\sb\omega\sp\ast\alpha\sp j\p\sb k\phi\sb\omega\,dx
=
i\int\psi\sb\omega\sp\ast\alpha\sp j\p\sb k\psi\sb\omega\,dx
=\delta^{j}\sb{k}L(\psi\sb\omega).
\end{equation}
Taking the trace of \eqref{t-q},
we obtain \emph{the Virial identity}
\begin{equation}\label{KL}
K(\phi\sb\omega)
=K(\psi\sb\omega)
=-n L(\psi\sb\omega),
\end{equation}
where
$K(\psi)$ is defined in \eqref{def-k-m-v}.

Note that for a solitary wave
$\psi\sb\omega(x,t)=\phi\sb\omega(x)e^{-i\omega t}$,
one has
$
E(\psi\sb\omega)=E(\phi\sb\omega)
$
since the Hamiltonian density
\eqref{nld-hamiltonian}
does not contain the time derivatives;
similarly,
the values of $K$, $M$, and $V$
are the same on $\psi\sb\omega$ and $\phi\sb\omega$.

\subsection*{Linearization at a solitary wave}

We assume that there are solitary wave solutions
to \eqref{nld}
of the form \eqref{sw}
with $\omega\in\Omega$,
where $\Omega$ is an open set.
Many quantities appearing below
will depend on $\omega$, which we will indicate
with the subscript $\omega$;
sometimes the subscript will be omitted to shorten the notations.

To study the linear stability
of the solitary waves \eqref{sw},
we consider the solution
$\psi$
in the form
\[
e^{-i\omega t}\big(\phi\sb\omega(x)+\rho(x,t)\big).
\]
The linearized equation is not $\C$-linear in $\rho$.
To apply the linear operator theory,
we write the linearized equation on $\rho$
in the $\C$-linear form
\[
\dot \bmuprho=\eubJ\eubL(\omega)\bmuprho,
\qquad
\bmuprho
=\begin{bmatrix}\Re\rho\\\Im\rho\end{bmatrix}
\]
with
\[
\eubJ=\begin{bmatrix}0&I\sb N\\-I\sb N&0\end{bmatrix},
\qquad
\eubL(\omega)=\eubD\sb m+\eubV,
\]
where
$
\eubD\sb m=\eubJ\bmupalpha\sp k\p\sb k+\bmupbeta m
$
and
\[
\bmupalpha\sp k=
\begin{bmatrix}\Re\alpha\sp k&-\Im\alpha\sp k
\\
\Im\alpha\sp k&\Re\alpha\sp k
\end{bmatrix}
\qquad
\bmupbeta=
\begin{bmatrix}\Re\beta&-\Im\beta
\\
\Im\beta&\Re\beta
\end{bmatrix}.
\]
The matrix-valued function
$\eubV$ is self-adjoint
and of Schwartz class in $x$;
its dependence on $\omega$
is via $\phi\sb\omega$.

\subsubsection*{The structure of the null space}

Due to the $\mathbf{U}(1)$-invariance of the equation, the
perturbation $\rho(x,t)$ that corresponds to infinitesimal multiplication
of the solitary wave by a constant unitary phase is in the kernel of
the linearization $\eubJ\eubL$.
Similarly, the translation invariance and the rotational symmetry
result in vectors in the kernel of the linearized operator.
As a result,
\begin{equation}
\eubJ\bmupphi\sb\omega,\,
\p\sb j\bmupphi\sb\omega\in\ker \eubJ\eubL(\omega),
\end{equation}
where
$
\bmupphi\sb\omega
=\begin{bmatrix}\Re\phi\sb\omega\\\Im\phi\sb\omega\end{bmatrix}.
$
These inclusions follow
from taking the derivatives in $\omega$ and $x\sp j$
of the relation $E'(\phi\sb\omega)=\omega Q'(\phi\sb\omega)$.
One can check by direct computation
that there is a Jordan block
corresponding to each of these eigenvectors:
\begin{equation}\label{xi-xi}
\eubJ\eubL\p\sb\omega\bmupphi\sb\omega=\eubJ\bmupphi\sb\omega,
\qquad
\eubJ\eubL\bmupxi\sb j=\p\sb j\bmupphi\sb\omega,
\end{equation}
where
\begin{equation}\label{def-xi}
\ \bmupxi\sb j=\omega x\sp j\eubJ\bmupphi\sb\omega
-\frac{1}{2}\bmupalpha\sp j\bmupphi\sb\omega.
\end{equation}

\medskip

By \eqref{xi-xi},
there are Jordan blocks of size at least $2$
corresponding to each of the vectors
$\eubJ\bmupphi\sb\omega$, $\p\sb j\bmupphi\sb\omega$
from the null space.
When two (or more)
eigenvalues collide at $\lambda=0$,
at a particular value of $\omega$,
they can
instantaneously join one of these two types of Jordan blocks permanently
residing at $0$.
We now consider these two events.

\subsubsection*{$\mathbf{U}(1)$-invariance
and Vakhitov--Kolokolov criterion}

Let us revisit the VK criterion
from the point of view of the size
of a particular Jordan block at $\lambda=0$.
By \eqref{xi-xi},
the Jordan block of $\eubJ\eubL$
corresponding to the unitary invariance
is of size at least $2$.
The size of this Jordan block
jumps up
when we can solve the generalized eigenvector equation
$\eubJ\eubL \bm u=\p\sb\omega\bmupphi\sb\omega$.
Since $\eubL$ is Fredholm
(this follows from $\eubL$ being self-adjoint
and $0\not\in\sigma_{\rm ess}(\eubL)=\R\backslash(-m,m)$;
see e.g. \cite{MR929030}),
such $\bm u$ exists
if $\p\sb\omega\bmupphi$ is orthogonal
to the null space of $(\eubJ\eubL)\sp\ast=-\eubL\eubJ$.
The generalized eigenvector
$\p\sb\omega\bmupphi$ is always orthogonal to
$\eubJ^{-1}\p\sb k\bmupphi\in\ker\eubL\eubJ$, $1\le k\le n$.
Indeed,
we have:
\begin{equation}\label{p-phi-j-phi-0}
\langle\p\sb\omega\bmupphi,\eubJ\p\sb k\bmupphi\rangle
=
-\langle\bmupphi\sb\omega,\bmupxi\sb k\rangle
=
\frac{\langle\bmupphi,\bmupalpha\sp k\bmupphi\rangle}{2}
-\omega\langle\bmupphi,x\sp k\eubJ\bmupphi\rangle,
\end{equation}
where we used \eqref{xi-xi}
and self-adjointness of $\eubL$.
By \eqref{j-zero},
the first term in the right-hand side is zero.
The second term in the right-hand side
is zero due to skew-symmetry of $\eubJ$.
Thus,
\begin{equation}\label{p-phi-j-phi}
\langle\p\sb\omega\bmupphi\sb\omega,\eubJ\p\sb k\bmupphi\sb\omega\rangle=0,
\qquad
1\le k\le n.
\end{equation}

We now need to check whether
$\p\sb\omega\bmupphi\sb\omega$ is orthogonal to
$\bmupphi\sb\omega\in\ker\eubL\eubJ$.
The orthogonality condition takes the form
\begin{equation}\label{condition-vk}
\langle\p\sb\omega\bmupphi\sb\omega,\bmupphi\sb\omega\rangle
=\frac 1 2\p\sb\omega Q(\phi\sb\omega)=0.
\end{equation}
This is in agreement
with the Vakhitov--Kolokolov criterion
$
\frac{d}{d\omega}Q(\phi\sb\omega)<0
$
derived in the context of the nonlinear Schr\"odinger equation
and more abstract Hamiltonian systems
with $\mathbf{U}(1)$-invariance \cite{VaKo,MR901236}.

\subsubsection*{Translation invariance and
the energy criterion of linear instability}

Let us find the condition
for the increase in size of the Jordan block corresponding
to translational invariance.
This happens if
there is $\bmupzeta$ such that
$\eubJ\eubL\bmupzeta=\bmupxi$,
where 
$\bmupxi=\sum\sb{j=1}\sp{n}c\sb j\bmupxi\sb j\ne 0$
is some nontrivial linear combination of generalized eigenvectors.
Since $\eubL$ is Fredholm,
the sufficient condition
is that $\bmupxi$
is orthogonal to vectors from $\ker\eubL\eubJ$.
By \eqref{p-phi-j-phi-0}
and \eqref{p-phi-j-phi}, one always has
\begin{equation}\label{orthtophi}
\langle\omega x^j\eubJ\bmupphi\sb\omega
-\frac 1 2\bmupalpha^j\bmupphi\sb\omega,
\bmupphi\sb\omega\rangle=0, 
\end{equation}
ensuring orthogonality of $\bmupxi$ to
$\bmupphi\sb\omega\in\ker\eubL\eubJ$.

Now we need to ensure orthogonality
to all of $\eubJ^{-1}\p\sb k\bmupphi\sb\omega\in\ker\eubL\eubJ$,
$1\le k\le n$.
We may write this condition
in the form
\begin{equation}\label{det-b-zero}
\det C\sb{jk}(\omega)=0,
\quad
C\sb{jk}(\omega):=-2\langle\bmupxi\sb j,\eubJ\p\sb k\bmupphi\sb\omega\rangle.
\end{equation}
Substituting $\bmupxi\sb j$ from \eqref{def-xi},
we have:
\begin{eqnarray}\label{a-phi-zero-0}
C\sb{jk}(\omega)
=
\langle\bmupalpha\sp j\bmupphi\sb\omega
-2\omega x\sp j\eubJ\bmupphi\sb\omega,
\eubJ\p\sb{k}\bmupphi\sb\omega\rangle.
\end{eqnarray}
Since
\[
\langle 2 x\sp j\eubJ\phi\sb\omega,
\eubJ\p\sb{k}\bmupphi\sb\omega\rangle
=\int x^j \p_k \left( \bmupphi\sb\omega^\ast\bmupphi\sb\omega\right)\,dx
=-\delta\sb k\sp j Q(\phi\sb\omega),
\]
we rewrite \eqref{a-phi-zero-0} as
\begin{eqnarray}\label{a-phi-zero}
C\sb{jk}(\omega)
=
\langle \bmupalpha^j\bmupphi\sb\omega, \eubJ\p_k\bmupphi\sb\omega \rangle
+\omega\delta\sb k\sp j Q(\omega).
\end{eqnarray}
Using
\eqref{l-h-q}, \eqref{t-q},
and
\eqref{a-phi-zero},
we get
$
C\sb{jk}(\omega)=E(\omega)\delta\sb{jk}.
$
Thus, the condition \eqref{det-b-zero}
for the increase of the size of the Jordan block
in the nonlinear Dirac equation
is equivalent to
\begin{equation}\label{CP-criterion}
E(\phi\sb\omega)=0.
\end{equation}

\begin{lemma}\label{lemma-ni-positive}
Let
$\mathcal{F}(\bar\psi,\psi)$
be homogeneous of degree $k+1$ in $\bar\psi$, $\psi$,
and assume that
$\mathcal{F}(\bar\psi,\psi)\ge 0$.
Then one has
\begin{equation}\label{e-positive}
E(\phi\sb\omega)>0
\qquad
\mbox{for}\quad
\omega>0.
\end{equation}
\end{lemma}

\begin{proof}
Substituting into
$\p\sb\lambda\at{\lambda=1}E(\phi\sb\lambda)
=\omega\p\sb\lambda\at{\lambda=1}Q(\phi\sb\lambda)$
the families
$\phi\sb\lambda(x)=\phi(x/\lambda)$
and $\phi\sb\lambda(x)=\lambda\phi(x)$,
we show that quantities
\eqref{def-k-m-v}
satisfy the relations
\[
\omega Q=\frac{n-1}{n}K+M+V,
\quad
\omega Q=K+M+(k+1)V.
\]
These relations yield $\frac 1 n K(\phi)=-k V(\phi)$.
With $V:=-\int \mathcal{F}(\bar\psi,\psi)\,dx<0$,
for $\omega>0$ we arrive at
$E=\omega Q-k V>0$.
\end{proof}

%
It follows that
in the pure power case
this instability mechanism
could only play the role
for the nonlinear Dirac solitary waves with $\omega<0$.

\subsubsection*{Rotational symmetry}

In the (3+1)D case,
the kernel of the linearized operator
contains the eigenvectors
due to the rotational symmetry.
For $1\le j\le 3$, denote
\[
\varSigma\sb j=\mathop{\rm diag}[\sigma\sb j,\sigma\sb j],
\quad
\bm\Sigma\sb j
=
\begin{bmatrix}\Re\varSigma\sb j&-\Im\varSigma\sb j\\
\Im\varSigma\sb j&\Re\varSigma\sb j\end{bmatrix}.
\]
Then
\begin{equation}\label{def-theta}
\bm\Theta_j=-\eubJ\bm\Sigma\sb j\bmupphi
+2 \epsilon_{j k l} x^k\p_l\bmupphi\sb\omega
\in\ker\eubJ\eubL(\omega)
\end{equation}
are the eigenvectors
from the null space
which correspond to infinitesimal rotations.
Above, $\epsilon\sb{j k l}$ are the Levi--Civita symbols.

It turns out that
$\eubJ\bmupphi\in\ker\eubJ\eubL$
is a linear combination
of $\bm\Theta\sb j$, $1\le j\le 3$,
so that these three eigenvectors
only contribute two into the dimension of $\ker\eubJ\eubL$.

One can check that
the condition for
the generalized eigenvector
$\p\sb\omega\bmupphi$
to be orthogonal to 
$\eubJ\bm\Theta_j$, $1\le k\le 3$,
is given by the VK condition
\[
\langle\p\sb\omega\bmupphi,\eubJ\bm\Theta_3\rangle
=\langle\p\sb\omega\bmupphi\sb\omega,\bmupphi\sb\omega\rangle.
\]
One can also check that
the generalized eigenvectors $\bmupxi\sb j$, $1\le j\le 3$,
are always orthogonal to $\bm\Theta\sb k$, $1\le k\le 3$:
\[
\Big\langle\omega x^j\eubJ\bmupphi
-\frac 1 2\alpha^j\bmupphi,
\eubJ\bm\Theta_k
\Big\rangle=0.
\]
Therefore,
the presence of these eigenvectors
in the kernel of $\eubJ\eubL$
in the (3+1)D case does not affect the
size of the Jordan blocks
associated with unitary and translational
invariance;
these sizes are completely characterized
by the conditions
\eqref{condition-vk} and \eqref{CP-criterion}.

There are no new Jordan blocks
associated to $\bm\Theta\sb j$.
For example, for the standard Ansatz
\begin{equation}\label{phi-is}
\phi\sb\omega(x)
=\begin{bmatrix}
g\sb\omega(r)\begin{pmatrix}1\\0\end{pmatrix}\\
if\sb\omega(r)
\begin{pmatrix}\cos\theta\\e^{i\varphi}\sin\theta\end{pmatrix}
\end{bmatrix},
\end{equation}
one has
$\bm\Theta\sb 3=-\eubJ\bmupphi$,
$\bm\Theta\sb 1=\eubJ\bm\Theta\sb 2$,
so that
\[
\langle\bm\Theta\sb 1,\eubJ\bm\Theta\sb 2\rangle
=\langle\bm\Theta\sb 1,\bm\Theta\sb 1\rangle
>0.
\]
As a consequence,
the Jordan block corresponding to
$\bm\Theta\sb 3$ is the same as the one
corresponding to the unitary invariance
(whose size is controlled by the VK condition \eqref{condition-vk}),
and there are no Jordan blocks corresponding to
$\bm\Theta\sb 1$, $\bm\Theta\sb 2$
since neither is orthogonal to $\ker(\eubJ\eubL)\sp\ast
\ni \eubJ^{-1}\bm\Theta\sb k$.

\begin{remark}\label{remark-2d}
In the $(2+1)$D case,
the story is similar:
the eigenvector from the null space
which corresponds to the infinitesimal rotation
coincides with $\eubJ\bmupphi$,
the same eigenvector which corresponds to
the unitary symmetry.
The size of the corresponding Jordan block
jumps (indicating collision of eigenvalues
at the origin)
if and only if the Vakhitov--Kolokolov condition
$\p\sb\omega Q(\phi\sb\omega)=0$
is satisfied.
\end{remark}



\section{Applications}
\label{sect-applications}

\subsection*{Generalized massive Thirring model}

The (generalized)
massive Thirring model in (1+1)D
is characterized by the Lagrangian
\begin{equation}\label{lagrange-mtm}
\mathscr{L}\sb{\rm MTM}
=\bar\psi(i\gamma\sp\mu\p\sb\mu-m)\psi
 +\frac{\norm{\bar\psi\gamma\sp\mu\psi}\sb g^{1+k}}{1+k},
\end{equation}
where
$\norm{\cdot}\sb g$
is the length in the Minkowski metric,
\[
\norm{\zeta}\sb g^2
=g\sb{\mu\nu}\zeta\sp\mu\zeta\sp\nu,
\qquad
\zeta\in\R^{1+1},
\]
with
$g=\mathop{\rm diag}[1,-1]$ the Minkowski tensor.
We notice that
$
\norm{\bar\psi\gamma\sp\mu\psi}\sb g^2
=(\psi\sp\ast\psi)^2-(\psi\sp\ast\alpha\sp 1\psi)^2\ge 0.
$

The choice $k=1$ leads to the nonlinear Dirac equations
with cubic nonlinearities
originally considered in  \cite{MR0091788}.
In the nonrelativistic limit
$\omega\lesssim m$,
for $k\in(0,2)$,
one has spectral stability
according to \cite{linear-a};
for $k>2$, there is linear instability
by \cite{MR3208458}.

There is an interesting behaviour for $\omega$ away
from the nonrelativistic limit.
It turns out that for any $k\ne 1$,
there is the following phenomenon:
there is $\omega\sb{E}=\omega\sb{E}(k)\in(-m,0)$
such that $E(\phi\sb\omega)\at{\omega=\omega\sb{E}}=0$.
According to our theory, at $\omega=\omega\sb{E}$,
two purely imaginary eigenvalues collide at the origin,
turning into a pair of two real (one positive, one negative)
eigenvalues for $\omega\in(-m,\omega\sb{E})$,
guaranteeing the linear instability in this region of frequencies.
On Figures~\ref{fig-mtm-0.5},
\ref{fig-mtm-1},
\ref{fig-mtm-2},
and~\ref{fig-mtm-3}
we plot the results of the numerical analysis
for the cases $k=1/2$, $1$, $2$, and $3$,
giving both the values of the energy and charge
functionals (as functions of $\omega\in(-m,m)$;
we take $m=1$)
and the spectrum of the operator
corresponding to the linearization
at a solitary wave.
We show that indeed the collision of eigenvalues at the origin
corresponds to either $\p\sb\omega Q(\phi\sb\omega)=0$
or to $E(\phi\sb\omega)=0$.

\subsection*{Gross--Neveu model}

The Soler model \cite{PhysRevD.1.2766}
has the Lagrange density
\begin{equation}\label{lagrange-soler}
\mathscr{L}\sb{\rm Soler}
=\bar\psi(i\gamma\sp\mu\p\sb\mu-m)\psi
+F(\bar\psi\psi).
\end{equation}
The corresponding equation is
\begin{equation}\label{soler}
\omega\psi
=(-i\bm\alpha\cdot\bm\nabla+m\beta)\psi
-f(\psi\sp\ast\beta\psi)\beta\psi,
\end{equation}
with $f(s)=F'(s)$.
It follows that if
$\phi\sb\omega(x)e^{-i\omega t}$
is a solitary wave solution
to the nonlinear Dirac equation \eqref{soler},
then $\phi\sb\omega(x)e^{+i\omega t}$
is a solitary wave solution to the nonlinear Dirac equation
\begin{equation}
i\dot\psi=
-i\hat{\bm\alpha}\cdot\bm\nabla
+(m-\hat f(\psi\sp\ast\hat\beta\psi))
\hat\beta\psi,
\end{equation}
with $\hat\alpha\sp j=-\alpha\sp j$,
$\hat\beta=-\beta$, and $\hat f(s)=f(-s)$.
Thus, one can always rewrite
a solitary wave solution with $\omega\in(-m,0)$
as a solitary wave
with $\omega\in(0,m)$.
Therefore, we conclude from
Lemma~\ref{lemma-ni-positive}
that for the Soler model
for pure power nonlinearities
$F(s)=\abs{s}^{k+1}/(k+1)$, $k>0$,
the condition $E(\phi\sb\omega)=0$
is not triggered;
the collisions of eigenvalues
at the origin
are described solely by the VK condition
$\frac{d}{d\omega}Q(\phi\sb\omega)=0$.
We provide the corresponding plots
on Figures~\ref{fig-gn-1} ($k=1/2$ and $k=1$)
and~\ref{fig-gn-2} ($k=2,\,3$).
One can see that the collision of eigenvalues
at the origin is indeed completely described by the
Vakhitov--Kolokolov condition $dQ/d\omega=0$.

\section{Discussion and conclusions}
We have shown that
generically the condition $E(\phi\sb\omega)=0$
indicates the birth of a pair of a positive and negative eigenvalues,
thus possibly marking the border of the linear instability region.
This condition is auxiliary to the Vakhitov--Kolokolov condition $\p\sb\omega Q(\phi\sb\omega)=0$
(see \cite{dirac-vk} on its application in the context of nonlinear Dirac equations).
Together, these two conditions describe the birth of real eigenvalues from the point $\lambda=0$,
as a result of a collision of a pair of purely imaginary eigenvalues.

The real eigenvalues produced after the
collision of eigenvalues at $\lambda=0$,
when $E(\phi\sb\omega)=0$,
correspond to vectors which are essentially
parallel to $\p\sb j\phi\sb\omega$.
Thus, the unstable behavior develops from a slight push,
after which a solitary wave
starts accelerating and loses its shape.

The condition of the energy vanishing
is of general type and is applicable to
any classical model of self-interacting spinors.
We have shown that
in the (generalized) massive Thirring model in (1+1)D
the energy vanishing
correctly predicts the eigenvalue collision,
and also have shown that this condition
is not triggered in the (generalized) Gross--Neveu model.

It can also be shown that for the NLS or Klein--Gordon equations, the
analogous criterion is never triggered
(each Jordan block corresponding to
translations is of size exactly two),
so the collision of eigenvalues at the origin
is completely described by the VK condition
$\p\sb\omega Q(\phi\sb\omega)=0$.
As a result, in the fermionic models,
the energy vanishing condition may take place
away from the nonrelativistic limit.
Indeed, we demonstrated that in the pure power case
the energy vanishing is only possible for solitary waves
with $\omega<0$.

It is important to mention that
this new criterion and the Vakhitov--Kolokolov
criterion do not exhaust all scenaria of instability.
In particular, two pairs of purely imaginary eigenvalues
may collide away from the origin and turn
into a quadruplet of complex eigenvalues with nonzero real part.
It is also possible that the instability
takes over
due to eigenvalues bifurcating from the essential spectrum.
We hope to address these scenaria in the forthcoming research.


\medskip

\noindent
{\bf Acknowledgments.}
We are grateful to Prof. V. Pokrovsky
and Dr. M. Zubkov
for helpful discussions
and two anonymous referees for helpful comments.

\newpage


\begin{figure}[htbp]

\hfill
\includegraphics[width=15cm,height=8cm]{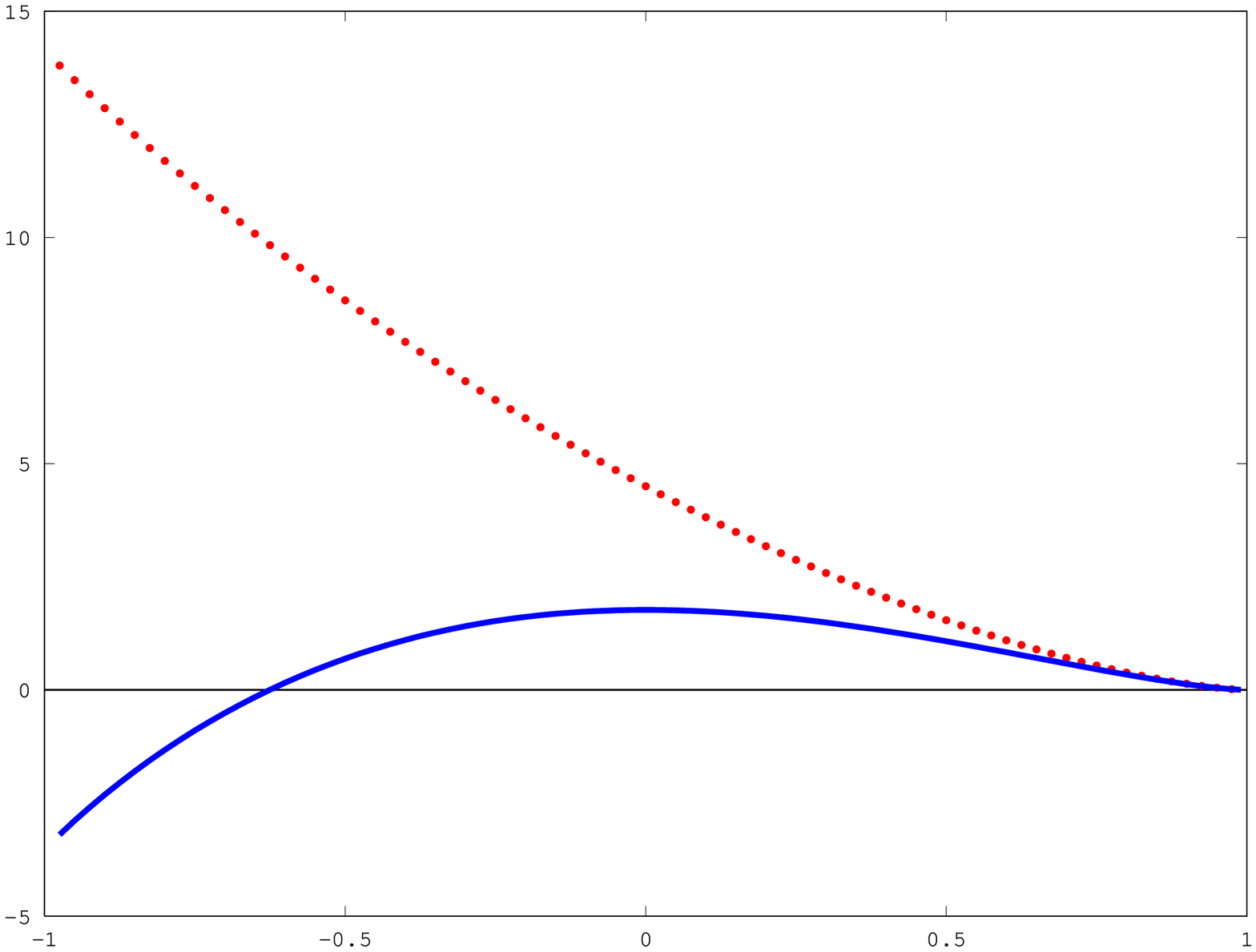}

\vskip -5mm

\hfill
\includegraphics[width=18cm,height=8cm]{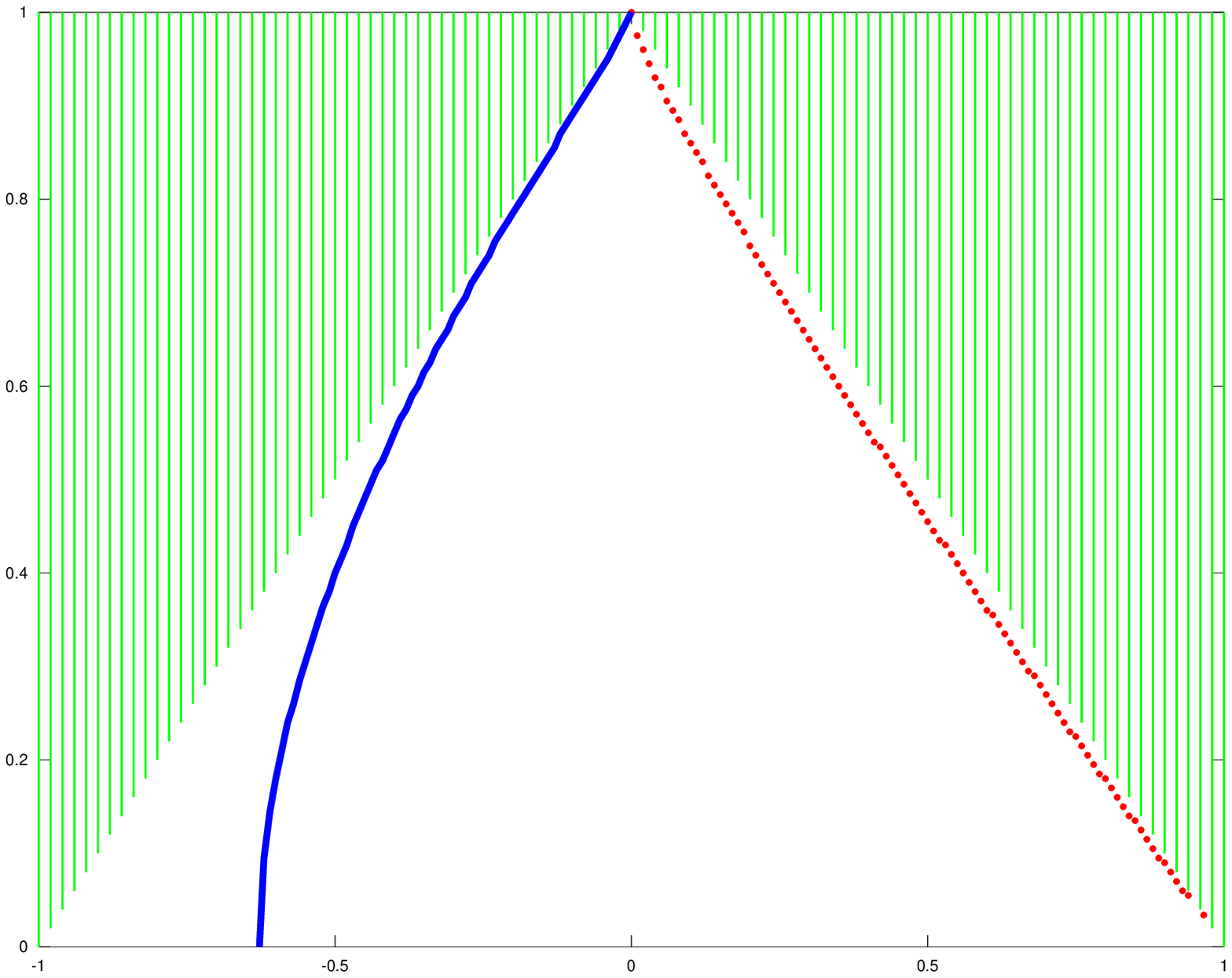}

\caption{\footnotesize
Massive Thirring model
with $k=1/2$.
\hfill\break
TOP: energy (solid line) and charge (dotted line)
as functions of $\omega\in(-1,1)$.
\hfill\break
BOTTOM:
The spectrum of the linearization at a solitary wave
on the upper half of the imaginary axis.
Solid vertical lines symbolize the (upper half of the) essential spectrum.
Dotted and solid curves denote ``even'' and ``odd'' eigenvalues
(of the same as $\phi$ and of the opposite parity, correspondingly).
The dotted eigenvalue
collides with its opposite at the origin
when $\omega=\omega\sb{E}\approx -0.6276$
(at $\omega\sb{E}$, the ``energy condition'' $E(\omega)=0$
is satisfied).
For $\omega\in(-1,\omega\sb{E})$, the spectrum contains
one positive and one negative eigenvalues (not shown).
}
\label{fig-mtm-0.5}
\end{figure}

\newpage

\begin{figure}[htbp]

\hfill
\includegraphics[width=15cm,height=8cm]{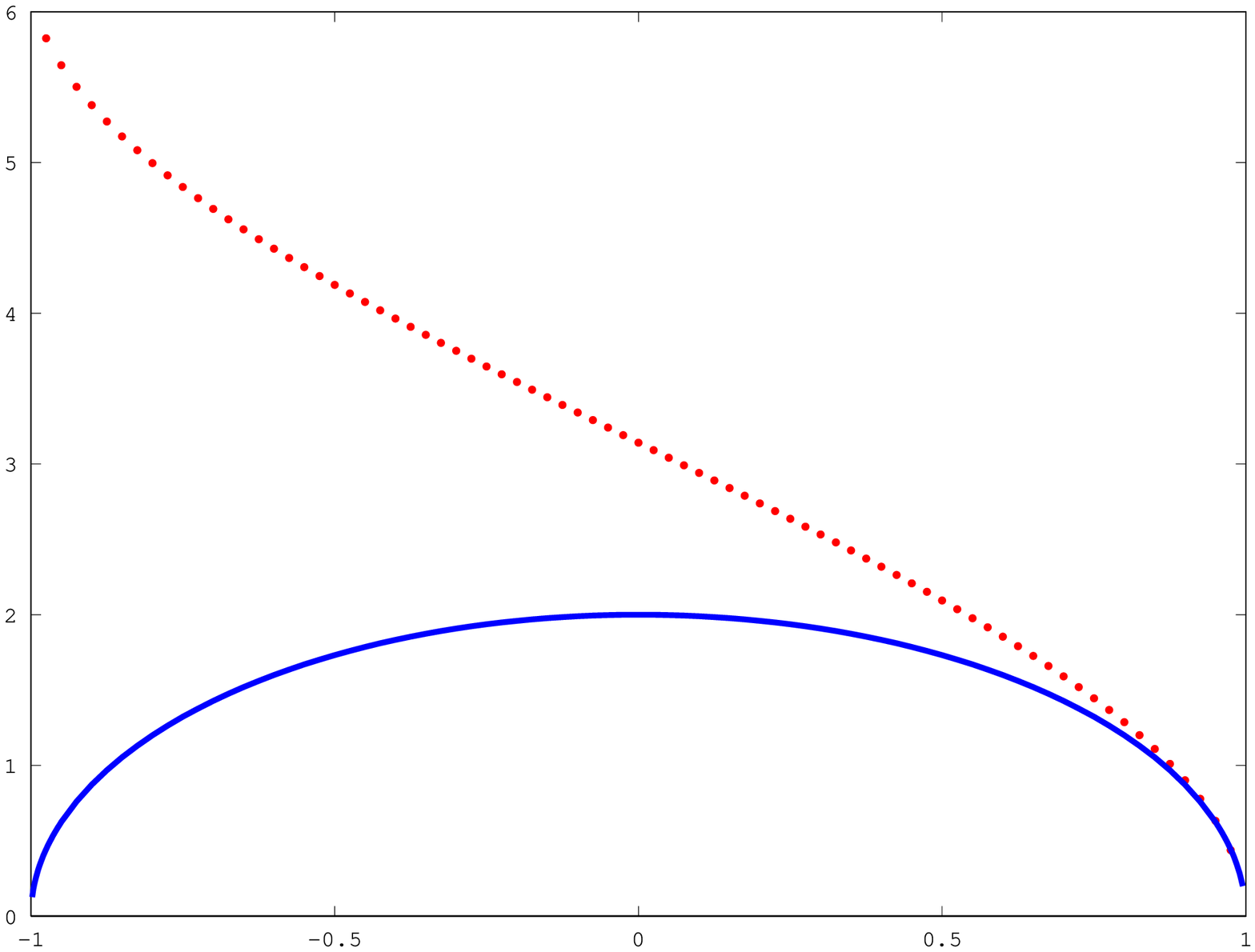}
\vskip -5mm

\hfill
\includegraphics[width=18cm,height=8cm]{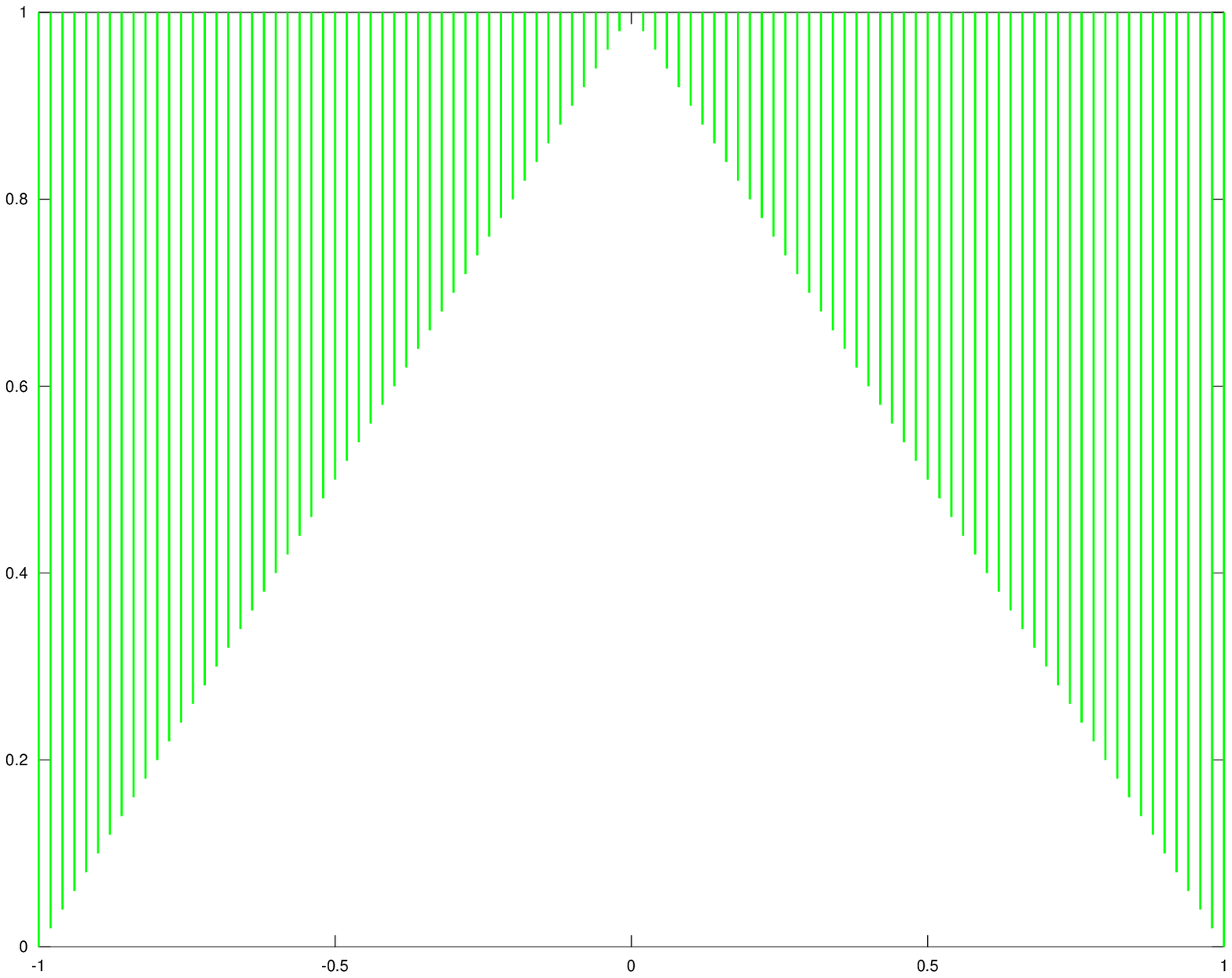}

\caption{\footnotesize
Massive Thirring model
with $k=1$.
\hfill\break
TOP: energy (solid line) and charge (dotted line)
as functions of $\omega\in(-1,1)$.
\hfill\break
BOTTOM:
The spectrum of the linearization at a solitary wave
on the upper half of the imaginary axis.
Note the absence of nonzero eigenvalues in the case
of the completely integrable model.
}

\label{fig-mtm-1}
\end{figure}
\newpage

\begin{figure}[htbp]

\hfill
\includegraphics[width=15cm,height=8cm]{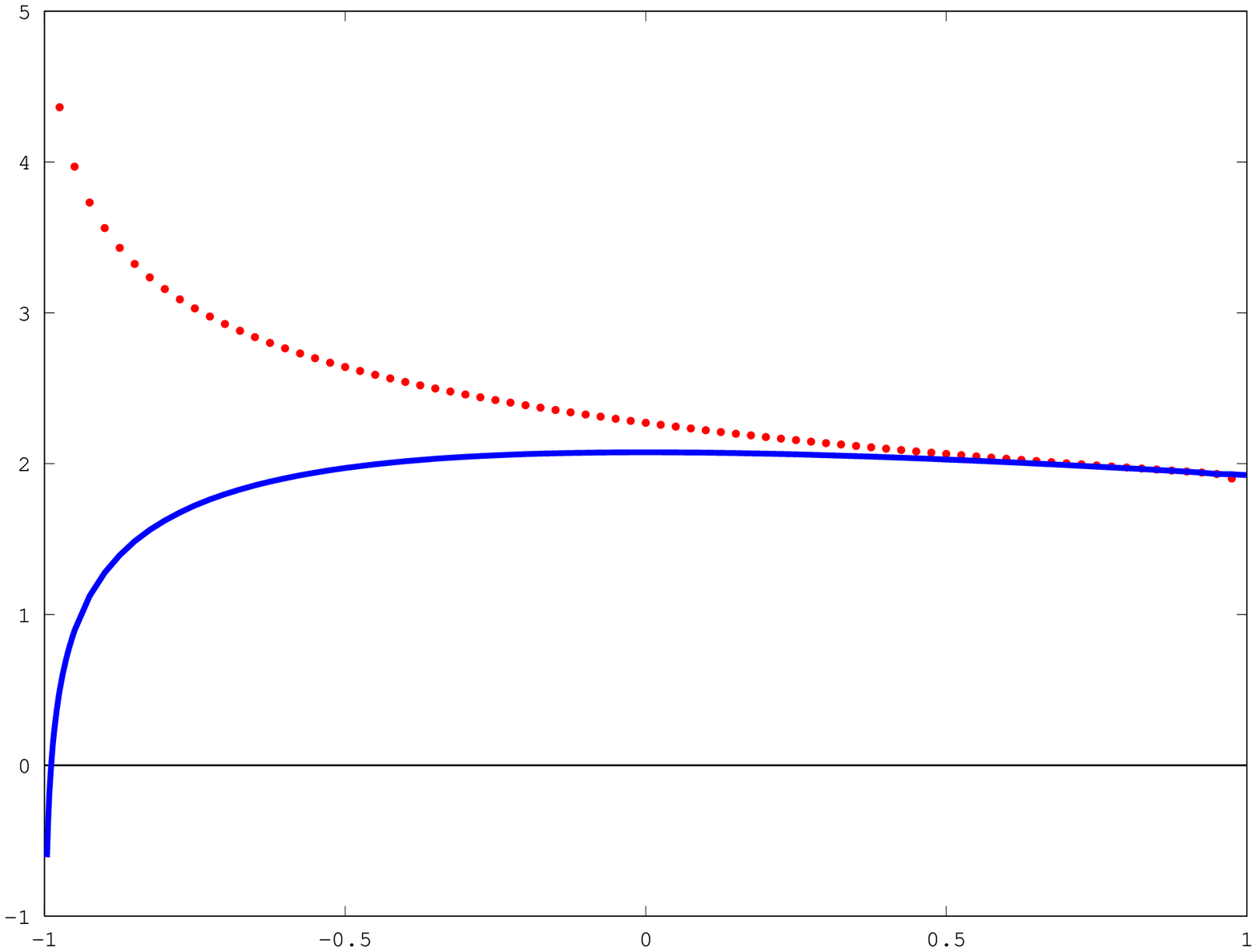}
\vskip -5mm

\hfill
\includegraphics[width=18cm,height=8cm]{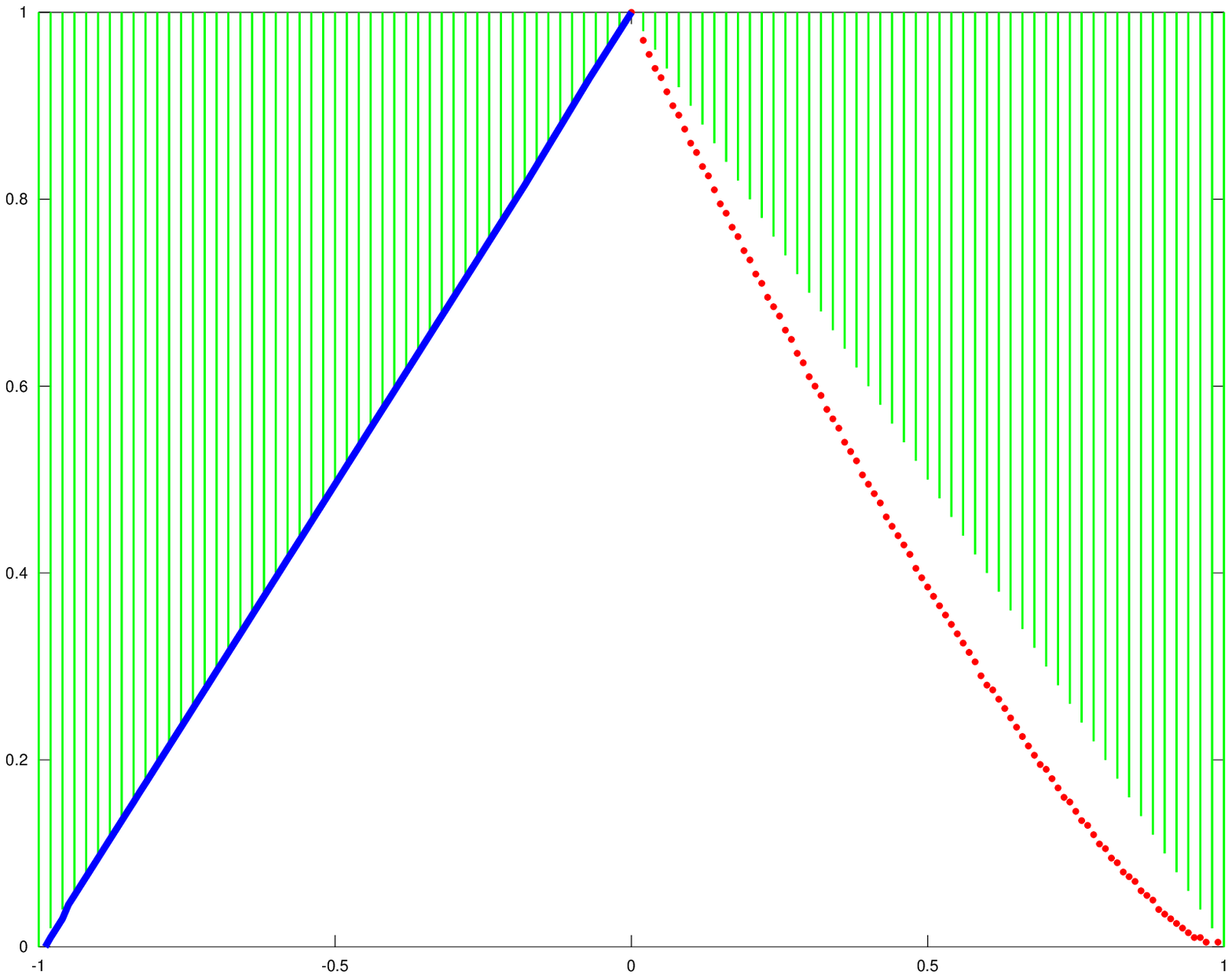}

\caption{\footnotesize
Massive Thirring model
with $k=2$.
\hfill\break
TOP: energy (solid line) and charge (dotted line)
as functions of $\omega\in(-1,1)$.
\hfill\break
BOTTOM:
The spectrum of the linearization at a solitary wave
on the upper half of the imaginary axis.
For $\omega\lesssim 1$,
note the presence of a purely imaginary eigenvalue
near $\lambda=0$
(dotted line)
whose trajectory is tangent
to the horizontal axis;
this is due to quintic NLS being charge-critical
in one spatial dimension.
\hfill\break
For $\omega\in(\omega\sb{E},0)$, there is a purely imaginary eigenvalue
near the threshold $\lambda=i(1-\abs{\omega})$.
It collides with its opposite at the origin
when $\omega=\omega\sb{E}\gtrsim-1$,
with $\omega\sb{E}$
corresponding to a solitary wave with zero energy.
For $\omega\in(-1,\omega\sb{E})$,
there is one positive and one negative eigenvalues
in the spectrum (not shown).
}
\label{fig-mtm-2}
\end{figure}
\newpage

\begin{figure}[htbp]

\hfill
\includegraphics[width=15cm,height=8cm]{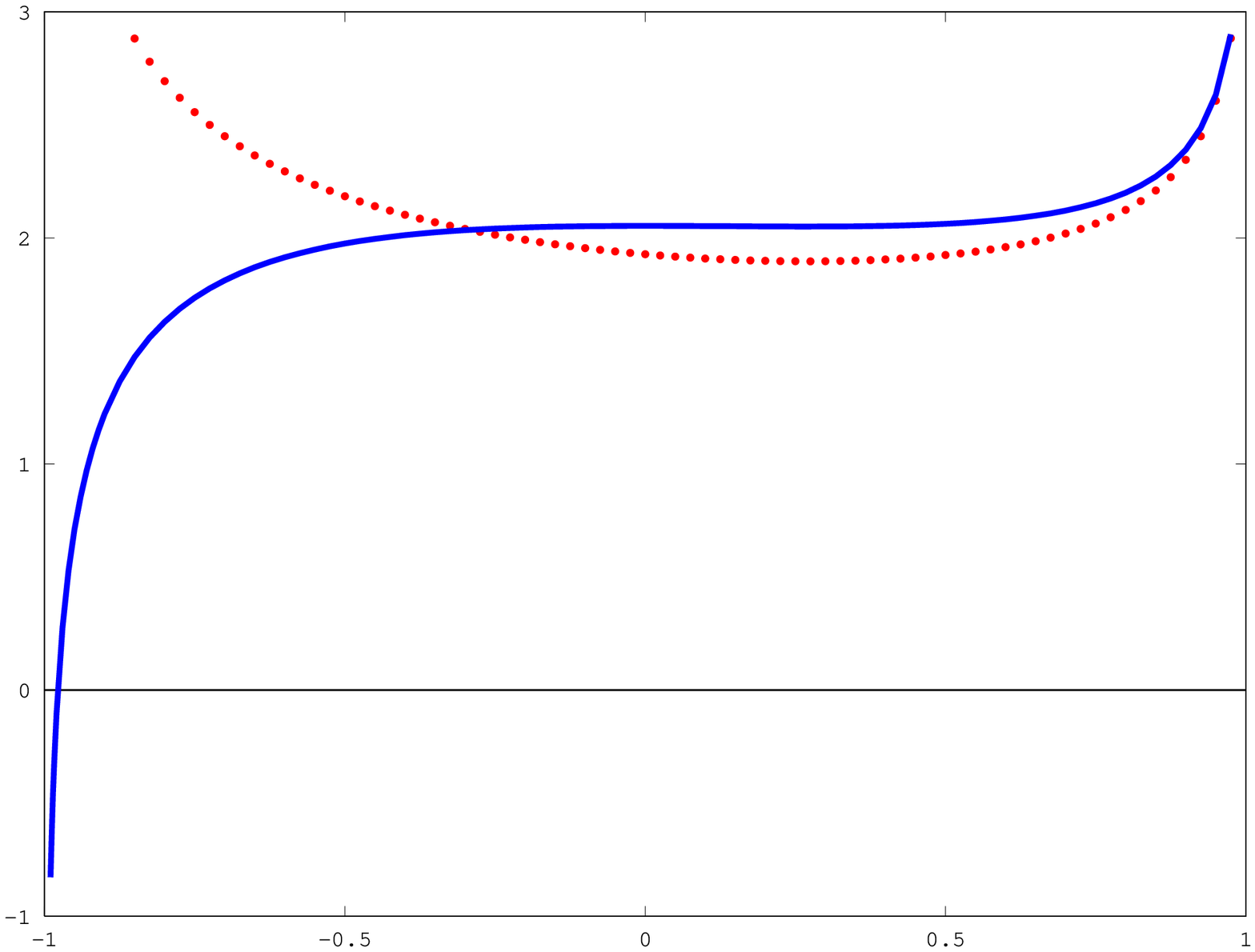}
\vskip -5mm

\hfill
\includegraphics[width=18cm,height=8cm]{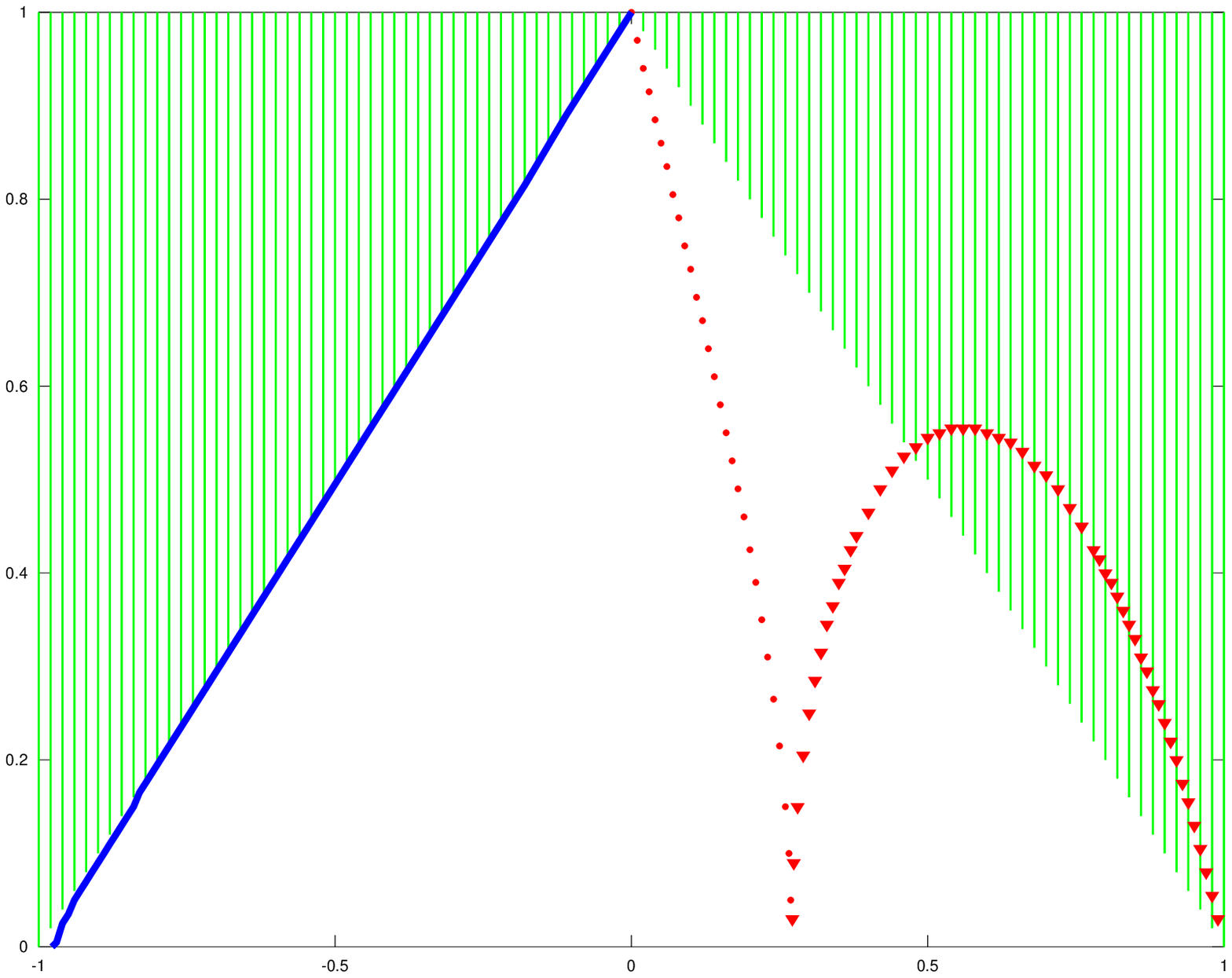}

\caption{\footnotesize
Massive Thirring model
with $k=3$.
\hfill\break
TOP: energy (solid line) and charge (dotted line)
as functions of $\omega\in(-1,1)$.
\hfill\break
BOTTOM:
The spectrum of the linearization at a solitary wave
on the upper half of the imaginary axis.
This case is charge-supercritical;
for $\omega\lesssim 1$,
there is a positive eigenvalue
(its trajectory is shown by triangles on the plot).
At $\omega=\omega\sb{\mathrm{VK}}$ (when the Vakhitov--Kolokolov condition
$dQ/d\omega=0$ is satisfied),
this real positive eigenvalue collides at the origin
with its opposite, producing a pair of purely imaginary eigenvalues
(the one with positive imaginary part is given by
the dotted line on the plot).
\hfill\break
For $\omega\in(\omega\sb{E},0)$, there is a purely imaginary eigenvalue
near the threshold $\lambda=i(1-\abs{\omega})$.
It collides with its opposite at the origin
when $\omega=\omega\sb{E}\gtrsim-1$,
where $\omega\sb{E}$ corresponds to a solitary wave
of zero energy.
For $\omega\in(-1,\omega\sb{E})$,
there is one positive and one negative eigenvalues
in the spectrum (not shown).
}

\label{fig-mtm-3}
\end{figure}

\newpage

\begin{figure}[htbp]

\includegraphics[width=8cm, height=8cm]{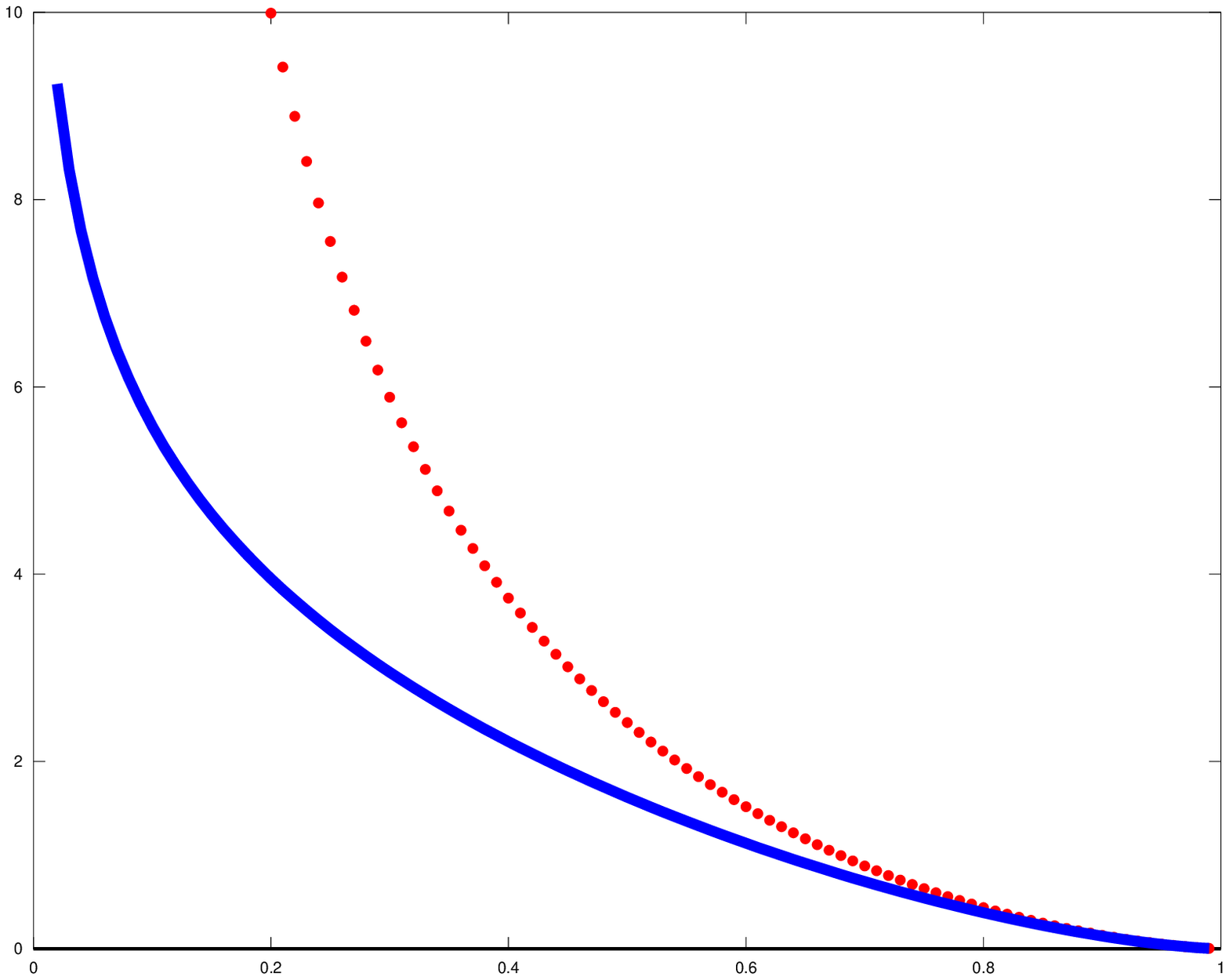}
\qquad
\includegraphics[width=8cm, height=8cm]{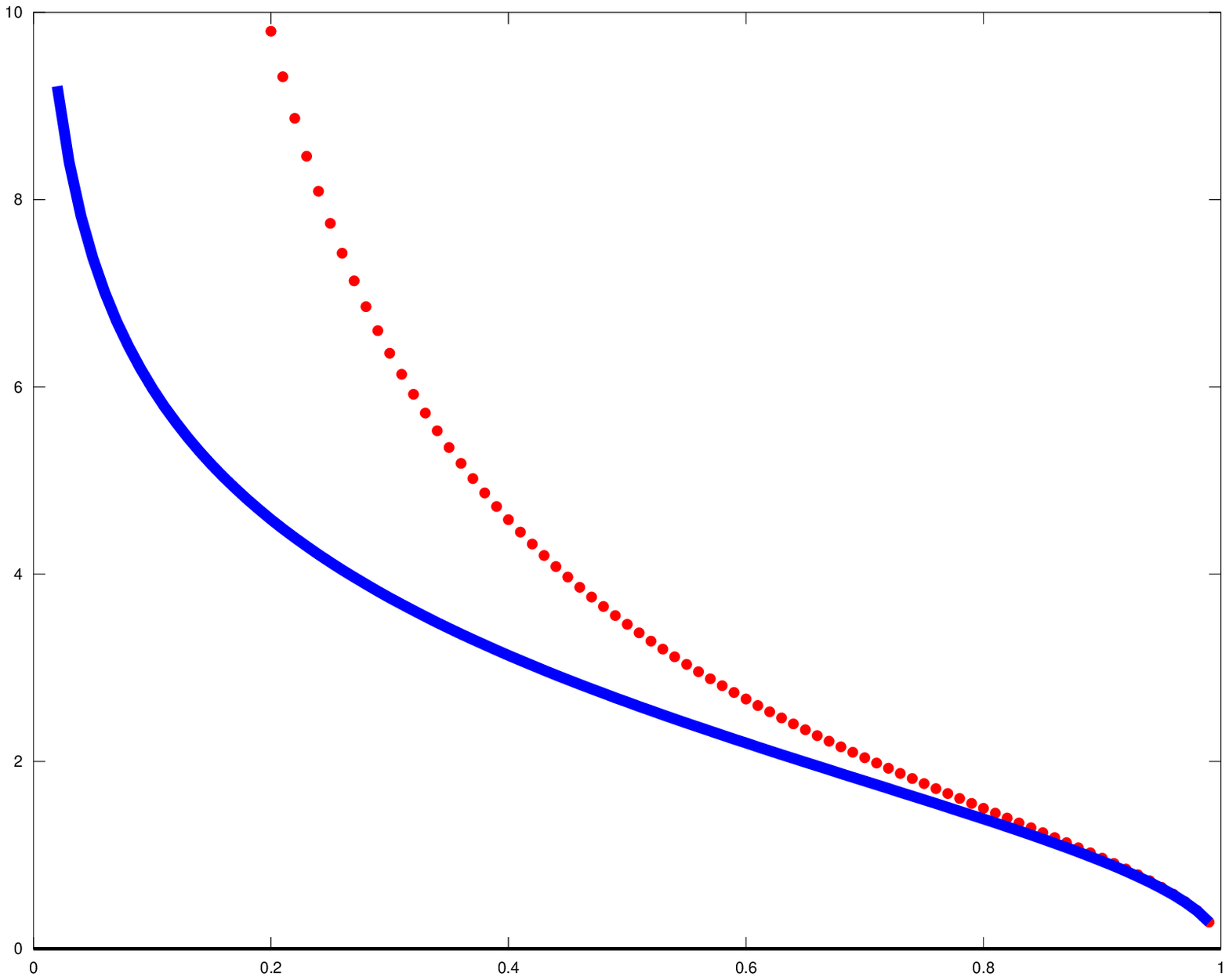}

\vskip -10mm

\includegraphics[width=8cm, height=8cm]{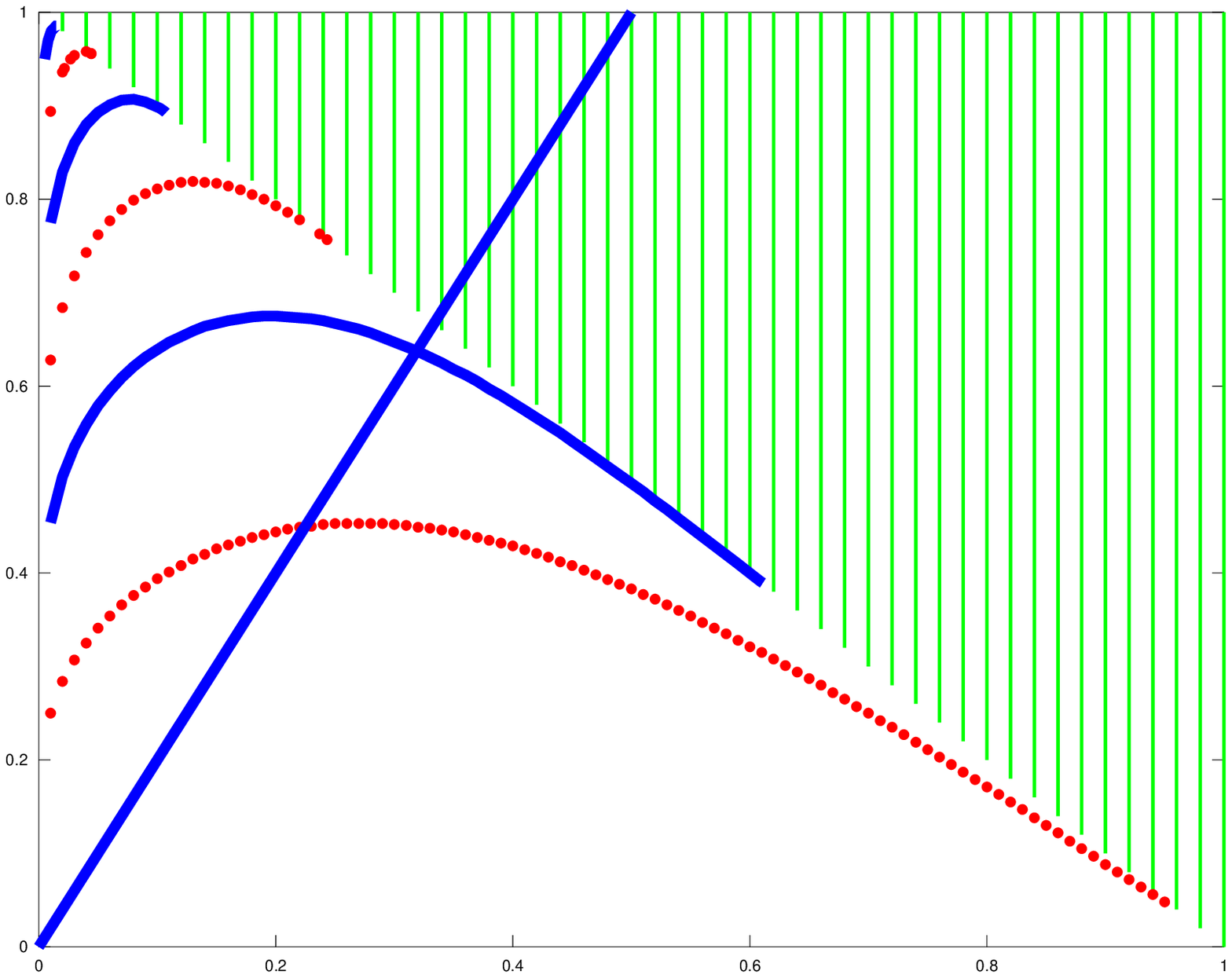}
\qquad
\includegraphics[width=8cm, height=8cm]{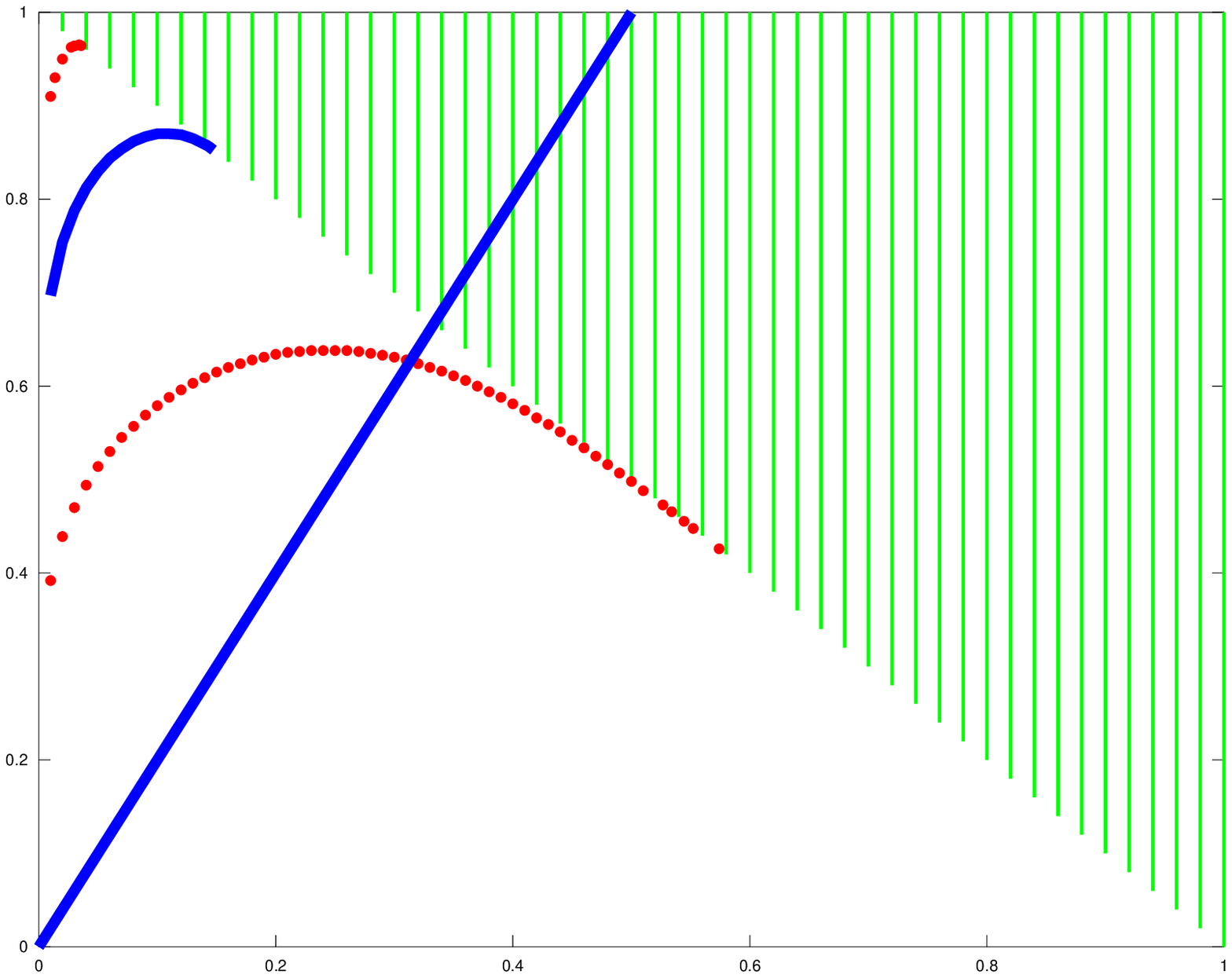}

\caption{\footnotesize
Gross--Neveu model. LEFT: $k=1/2$; RIGHT: $k=1$.
\hfill\break
TOP ROW:
charge (dotted line) and energy (solid line)
of the solitary waves
as functions of $\omega\in(0,1)$.
\hfill\break
BOTTOM ROW:
Spectrum on the upper half of the imaginary axis.
On each of the plots
in the bottom raw,
note the exact eigenvalue $\lambda=2\omega i$
\cite{dirac-vk}.
We already presented
the spectrum corresponding to $k=1$
(with a detailed description of our numerical methods)
in \cite{MR2892774}.
}

\label{fig-gn-1}
\end{figure}

\begin{figure}[htbp]

\includegraphics[width=8cm, height=8cm]{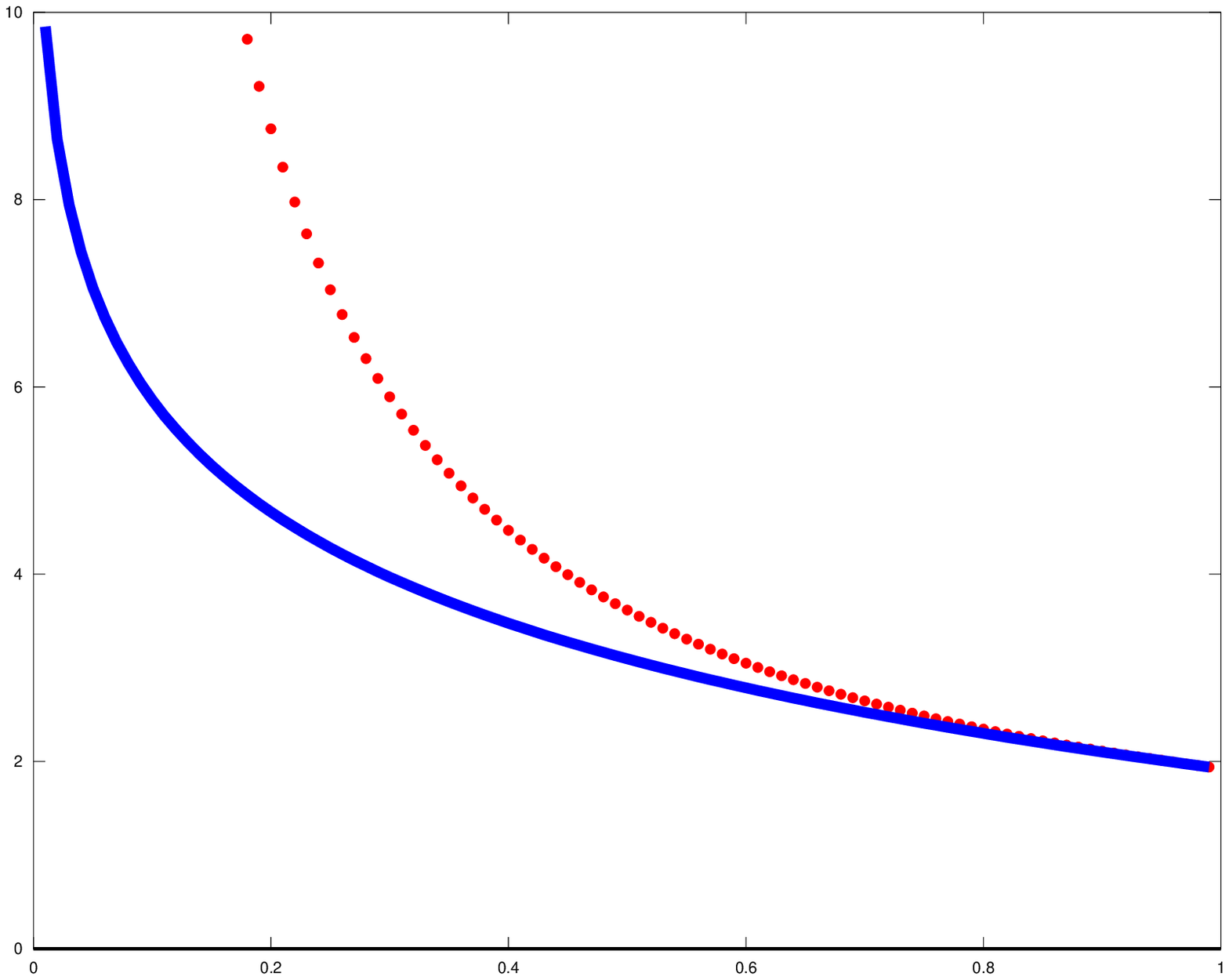}
\qquad
\includegraphics[width=8cm, height=8cm]{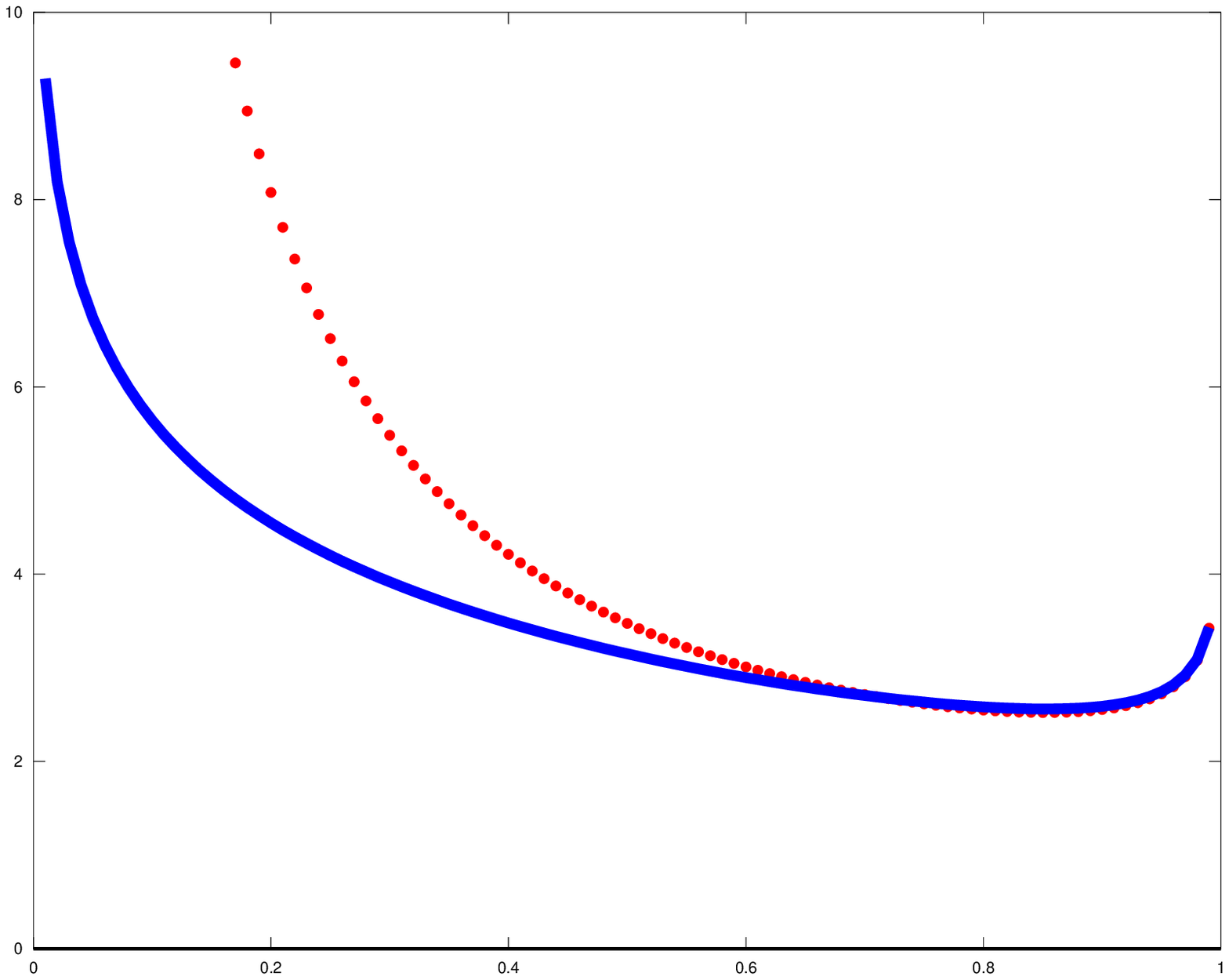}

\vskip -10mm

\includegraphics[width=8cm, height=8cm]{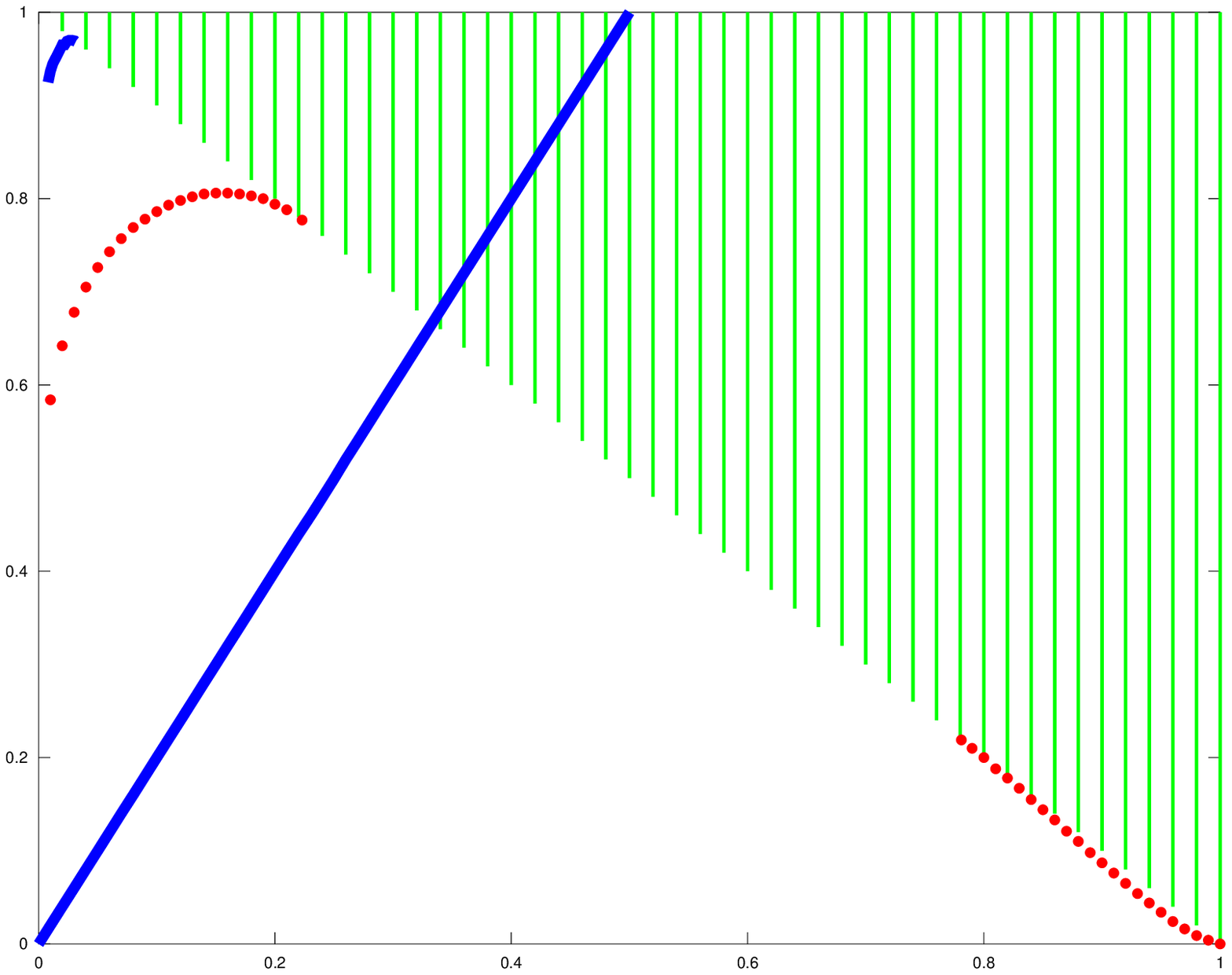}
\qquad
\includegraphics[width=8cm, height=8cm]{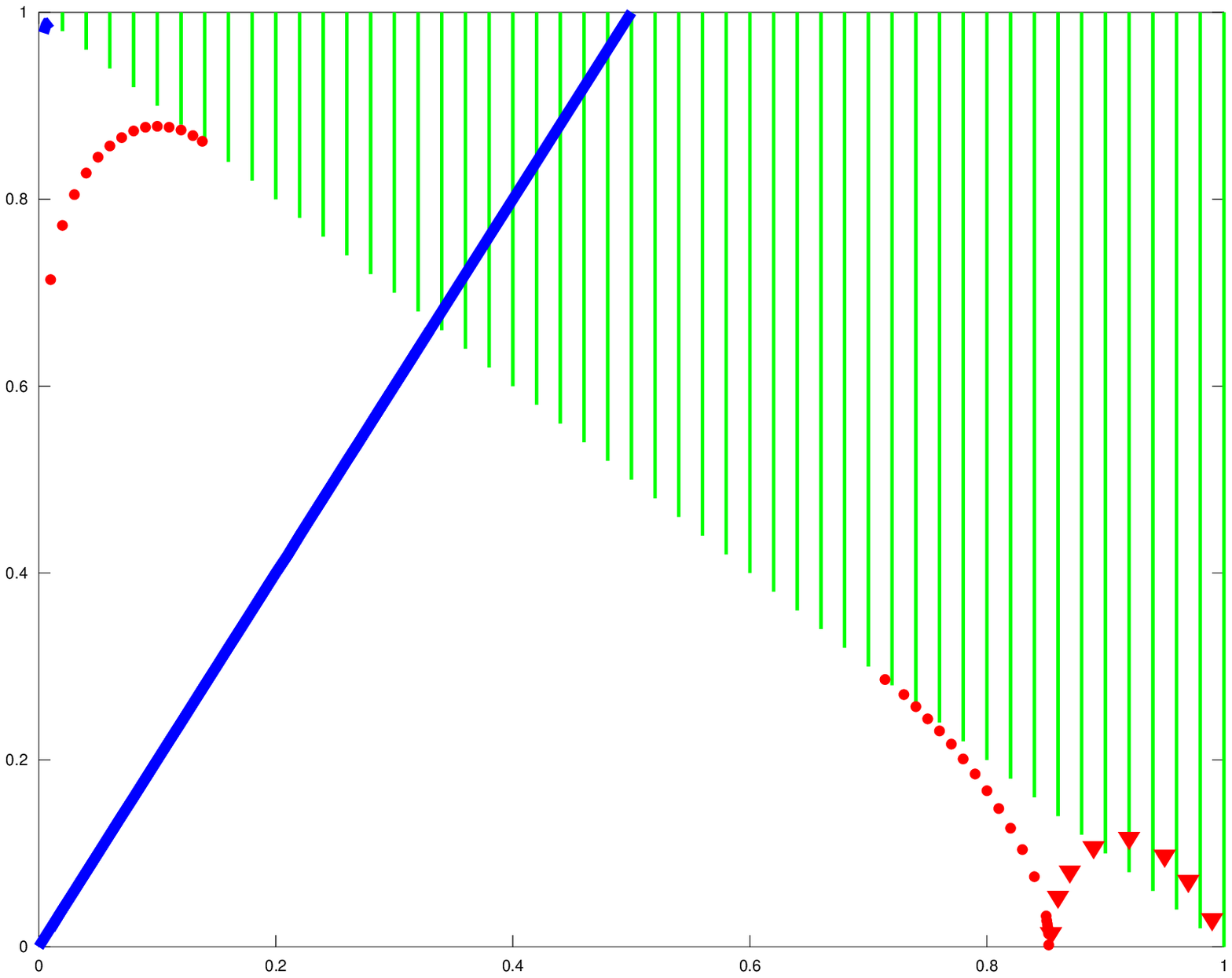}

\caption{\footnotesize
Gross--Neveu model. LEFT: $k=2$; RIGHT: $k=3$.
\hfill\break
TOP ROW:
charge (dotted line) and energy (solid line)
of the solitary waves
as functions of $\omega\in(0,1)$.
\hfill\break
BOTTOM ROW:
Spectrum on the upper half of the imaginary axis.
\hfill\break
The case $k=2$ is charge-critical;
note the purely imaginary eigenvalue whose trajectory
is tangent to $\lambda=0$
for $\omega\lesssim 1$.
\hfill\break
The case $k=3$ is charge-supercritical;
there is a real eigenvalue
(plotted with triangles)
born from the origin at $\omega=1$
and persisting $\omega\lesssim 1$.
At $\omega=\omega\sb{\mathrm{VK}}$ (when the Vakhitov--Kolokolov
condition $dQ/d\omega=0$ is satisfied),
this eigenvalue collides with its opposite
at the origin, producing a pair of purely imaginary
eigenvalues (one with the positive imaginary part
is plotted with dots).
}

\label{fig-gn-2}

\end{figure}

\clearpage
\newpage

\bibliography{all}

\def\cprime{$'$} \def\polhk#1{\setbox0=\hbox{#1}{\ooalign{\hidewidth
  \lower1.5ex\hbox{`}\hidewidth\crcr\unhbox0}}} \def\cprime{$'$}
\begin{thebibliography}{CKMS10}

\bibitem[Abe98]{MR1618672}
S.~Abenda, \href{http://www.numdam.org/item?id=AIHPA_1998__68_2_229_0}{{\em
  Solitary waves for {M}axwell-{D}irac and {C}oulomb-{D}irac models\/}}, Ann.
  Inst. H. Poincar\'e Phys. Th\'eor. {\bf 68} (1998), pp. 229--244.

\bibitem[AS83]{PhysRevLett.50.1230}
A.~Alvarez and M.~Soler,
  \href{http://link.aps.org/doi/10.1103/PhysRevLett.50.1230}{{\em Energetic
  stability criterion for a nonlinear spinorial model\/}}, Phys. Rev. Lett.
  {\bf 50} (1983), pp. 1230--1233.

\bibitem[AS86]{PhysRevD.34.644}
A.~Alvarez and M.~Soler,
  \href{http://link.aps.org/doi/10.1103/PhysRevD.34.644}{{\em Stability of the
  minimum solitary wave of a nonlinear spinorial model\/}}, Phys. Rev. D {\bf
  34} (1986), pp. 644--645.

\bibitem[BC12a]{MR2892774}
G.~Berkolaiko and A.~Comech,
  \href{http://dx.doi.org/10.1051/mmnp/20127202}{{\em On spectral stability of
  solitary waves of nonlinear {D}irac equation in 1{D}\/}}, Math. Model. Nat.
  Phenom. {\bf 7} (2012), pp. 13--31.

\bibitem[BC12b]{linear-a}
N.~Boussaid and A.~Comech, \href{http://arxiv.org/abs/1211.3336}{{\em On
  spectral stability of nonlinear {D}irac equation\/}}, ArXiv e-prints  (2012).

\bibitem[BD65]{MR0187642}
J.~D. Bjorken and S.~D. Drell, {\em Relativistic quantum fields\/},
  McGraw-Hill, New York, 1965.

\bibitem[Bog79]{MR592382}
I.~L. Bogolubsky, \href{http://dx.doi.org/10.1016/0375-9601(79)90442-0}{{\em On
  spinor soliton stability\/}}, Phys. Lett. A {\bf 73} (1979), pp. 87--90.

\bibitem[BPZ98]{PhysRevLett.80.5117}
I.~V. Barashenkov, D.~E. Pelinovsky, and E.~V. Zemlyanaya,
  \href{http://link.aps.org/doi/10.1103/PhysRevLett.80.5117}{{\em Vibrations
  and oscillatory instabilities of gap solitons\/}}, Phys. Rev. Lett. {\bf 80}
  (1998), pp. 5117--5120.

\bibitem[BSV87]{PhysRevD.36.2422}
P.~Blanchard, J.~Stubbe, and L.~V\`azquez,
  \href{http://link.aps.org/doi/10.1103/PhysRevD.36.2422}{{\em {Stability of
  nonlinear spinor fields with application to the {G}ross-{N}eveu model}\/}},
  Phys. Rev. D {\bf 36} (1987), pp. 2422--2428.

\bibitem[CGG14]{MR3208458}
A.~Comech, M.~Guan, and S.~Gustafson,
  \href{http://dx.doi.org/10.1016/j.anihpc.2013.06.001}{{\em On linear
  instability of solitary waves for the nonlinear {D}irac equation\/}}, Ann.
  Inst. H. Poincar\'e Anal. Non Lin\'eaire {\bf 31} (2014), pp. 639--654.

\bibitem[CKMS10]{PhysRevE.82.036604}
F.~Cooper, A.~Khare, B.~Mihaila, and A.~Saxena,
  \href{http://dx.doi.org/10.1103/PhysRevE.82.036604}{{\em Solitary waves in
  the nonlinear {D}irac equation with arbitrary nonlinearity\/}}, Phys. Rev. E
  {\bf 82} (2010), p. 036604.

\bibitem[CKS14]{2014arXiv1408.4171C}
J.~{Cuevas-Maraver}, P.~G. {Kevrekidis}, and A.~{Saxena},
  \href{http://arxiv.org/abs/1408.4171}{{\em {Solitary Waves in a Discrete
  Nonlinear Dirac equation}\/}}, ArXiv e-prints  (2014).

\bibitem[Com11]{dirac-vk}
A.~Comech, \href{http://arxiv.org/abs/1107.1763}{{\em {On the meaning of the
  {V}akhitov-{K}olokolov stability criterion for the nonlinear {D}irac
  equation}\/}}, ArXiv e-prints  (2011).

\bibitem[CP06]{MR2217129}
M.~Chugunova and D.~Pelinovsky, \href{http://dx.doi.org/10.1137/050629781}{{\em
  Block-diagonalization of the symmetric first-order coupled-mode system\/}},
  SIAM J. Appl. Dyn. Syst. {\bf 5} (2006), pp. 66--83.

\bibitem[CPS13]{2013arXiv1312.1019C}
A.~{Contreras}, D.~E. {Pelinovsky}, and Y.~{Shimabukuro},
  \href{http://arxiv.org/abs/1312.1019}{{\em {\$L\^{}2\$ orbital stability of
  Dirac solitons in the massive Thirring model}\/}}, ArXiv e-prints  (2013).

\bibitem[CS12]{dm-existence}
A.~Comech and D.~Stuart, \href{http://arxiv.org/abs/1210.7261}{{\em Small
  solitary waves in the {D}irac-{M}axwell system\/}}, ArXiv e-prints  (2012).

\bibitem[CZ13]{MR3118823}
A.~Comech and M.~Zubkov,
  \href{http://dx.doi.org/10.1088/1751-8113/46/43/435201}{{\em Polarons as
  stable solitary wave solutions to the {D}irac-{C}oulomb system\/}}, J. Phys.
  A {\bf 46} (2013), pp. 435201, 21.

\bibitem[DHN74]{PhysRevD.10.4130}
R.~F. Dashen, B.~Hasslacher, and A.~Neveu,
  \href{http://link.aps.org/doi/10.1103/PhysRevD.10.4130}{{\em Nonperturbative
  methods and extended-hadron models in field theory. ii. two-dimensional
  models and extended hadrons\/}}, Phys. Rev. D {\bf 10} (1974), pp.
  4130--4138.

\bibitem[EE87]{MR929030}
D.~E. Edmunds and W.~D. Evans, {\em Spectral theory and differential
  operators\/}, Oxford Mathematical Monographs, The Clarendon Press Oxford
  University Press, New York, 1987, oxford Science Publications.

\bibitem[EGS96]{MR1386737}
M.~J. Esteban, V.~Georgiev, and E.~S{\'e}r{\'e},
  \href{http://dx.doi.org/10.1007/BF01254347}{{\em Stationary solutions of the
  {M}axwell-{D}irac and the {K}lein-{G}ordon-{D}irac equations\/}}, Calc. Var.
  Partial Differential Equations {\bf 4} (1996), pp. 265--281.

\bibitem[FFK56]{PhysRev.103.1571}
R.~Finkelstein, C.~Fronsdal, and P.~Kaus,
  \href{http://link.aps.org/doi/10.1103/PhysRev.103.1571}{{\em Nonlinear spinor
  field\/}}, Phys. Rev. {\bf 103} (1956), pp. 1571--1579.

\bibitem[FLR51]{PhysRev.83.326}
R.~Finkelstein, R.~LeLevier, and M.~Ruderman,
  \href{http://link.aps.org/doi/10.1103/PhysRev.83.326}{{\em Nonlinear spinor
  fields\/}}, Phys. Rev. {\bf 83} (1951), pp. 326--332.

\bibitem[FSY99]{PhysRevD.59.104020}
F.~Finster, J.~Smoller, and S.-T. Yau,
  \href{http://dx.doi.org/10.1103/PhysRevD.59.104020}{{\em Particlelike
  solutions of the {E}instein-{D}irac equations\/}}, Phys. Rev. D {\bf 59}
  (1999), p. 104020.

\bibitem[GN74]{PhysRevD.10.3235}
D.~J. Gross and A.~Neveu,
  \href{http://dx.doi.org/10.1103/PhysRevD.10.3235}{{\em Dynamical symmetry
  breaking in asymptotically free field theories\/}}, Phys. Rev. D {\bf 10}
  (1974), pp. 3235--3253.

\bibitem[GSS87]{MR901236}
M.~Grillakis, J.~Shatah, and W.~Strauss,
  \href{http://dx.doi.org/10.1016/0022-1236(87)90044-9}{{\em Stability theory
  of solitary waves in the presence of symmetry. {I}\/}}, J. Funct. Anal. {\bf
  74} (1987), pp. 160--197.

\bibitem[GSS90]{MR1081647}
M.~Grillakis, J.~Shatah, and W.~Strauss,
  \href{http://dx.doi.org/10.1016/0022-1236(90)90016-E}{{\em Stability theory
  of solitary waves in the presence of symmetry. {II}\/}}, J. Funct. Anal. {\bf
  94} (1990), pp. 308--348.

\bibitem[GW08]{MR2513792}
R.~H. Goodman and M.~I. Weinstein,
  \href{http://dx.doi.org/10.1016/j.physd.2008.04.009}{{\em Stability and
  instability of nonlinear defect states in the coupled mode
  equations---analytical and numerical study\/}}, Phys. D {\bf 237} (2008), pp.
  2731--2760.

\bibitem[Hei57]{RevModPhys.29.269}
W.~Heisenberg, \href{http://link.aps.org/doi/10.1103/RevModPhys.29.269}{{\em
  Quantum theory of fields and elementary particles\/}}, Rev. Mod. Phys. {\bf
  29} (1957), pp. 269--278.

\bibitem[Iva38]{jetp.8.260}
D.~D. Ivanenko, {\em Notes to the theory of interaction via particles\/}, Zh.
  \'Eksp. Teor. Fiz {\bf 8} (1938), pp. 260--266.

\bibitem[KKS04]{MR2089513}
T.~Kapitula, P.~G. Kevrekidis, and B.~Sandstede,
  \href{http://dx.doi.org/10.1016/j.physd.2004.03.018}{{\em Counting
  eigenvalues via the {K}rein signature in infinite-dimensional {H}amiltonian
  systems\/}}, Phys. D {\bf 195} (2004), pp. 263--282.

\bibitem[KP12]{MR2869071}
T.~Kapitula and K.~Promislow,
  \href{http://dx.doi.org/10.1090/S0002-9939-2011-10943-2}{{\em Stability
  indices for constrained self-adjoint operators\/}}, Proc. Amer. Math. Soc.
  {\bf 140} (2012), pp. 865--880.

\bibitem[Lan33]{polaron-1933}
L.~D. Landau, Phys. Z. Sowjetunion {\bf 3} (1933), p. 664.

\bibitem[LG75]{PhysRevD.12.3880}
S.~Y. Lee and A.~Gavrielides,
  \href{http://dx.doi.org/10.1103/PhysRevD.12.3880}{{\em Quantization of the
  localized solutions in two-dimensional field theories of massive
  fermions\/}}, Phys. Rev. D {\bf 12} (1975), pp. 3880--3886.

\bibitem[Lis95]{MR1364144}
A.~G. Lisi, \href{http://stacks.iop.org/0305-4470/28/5385}{{\em A solitary wave
  solution of the {M}axwell-{D}irac equations\/}}, J. Phys. A {\bf 28} (1995),
  pp. 5385--5392.

\bibitem[LL84]{MR766230}
L.~D. Landau and E.~M. Lifshitz, {\em Course of theoretical physics. {V}ol.
  8\/}, Pergamon International Library of Science, Technology, Engineering and
  Social Studies, Pergamon Press, Oxford, 1984, electrodynamics of continuous
  media, Translated from the second Russian edition by J. B. Sykes, J. S. Bell
  and M. J. Kearsley, Second Russian edition revised by Lifshits and L. P.
  Pitaevski{\u\i}.

\bibitem[Pek46]{polaron-1946}
S.~Pekar, Zh. \'Eksp. Teor. Fiz. {\bf 16} (1946), p. 341.

\bibitem[Pel05]{MR2121936}
D.~E. Pelinovsky, \href{http://dx.doi.org/10.1098/rspa.2004.1345}{{\em Inertia
  law for spectral stability of solitary waves in coupled nonlinear
  {S}chr\"odinger equations\/}}, Proc. R. Soc. Lond. Ser. A Math. Phys. Eng.
  Sci. {\bf 461} (2005), pp. 783--812.

\bibitem[PS14]{2014LMaPh.104...21P}
D.~E. {Pelinovsky} and Y.~{Shimabukuro},
  \href{http://dx.doi.org/10.1007/s11005-013-0650-5}{{\em {Orbital Stability of
  Dirac Solitons}\/}}, Letters in Mathematical Physics {\bf 104} (2014), pp.
  21--41.

\bibitem[PSK04]{PhysRevE.70.036618}
D.~E. Pelinovsky, A.~A. Sukhorukov, and Y.~S. Kivshar,
  \href{http://link.aps.org/doi/10.1103/PhysRevE.70.036618}{{\em Bifurcations
  and stability of gap solitons in periodic potentials\/}}, Phys. Rev. E {\bf
  70} (2004), p. 036618.

\bibitem[QGW04]{QUA:QUA20146}
H.~M. Quiney, V.~N. Glushkov, and S.~Wilson,
  \href{http://dx.doi.org/10.1002/qua.20146}{{\em The dirac equation in the
  algebraic approximation. ix. {M}atrix {D}irac--{H}artree--{F}ock calculations
  for the {HeH} and {BeH} ground states using distributed gaussian basis
  sets\/}}, International Journal of Quantum Chemistry {\bf 99} (2004), pp.
  950--962.

\bibitem[Sol70]{PhysRevD.1.2766}
M.~Soler, \href{http://dx.doi.org/10.1103/PhysRevD.1.2766}{{\em Classical,
  stable, nonlinear spinor field with positive rest energy\/}}, Phys. Rev. D
  {\bf 1} (1970), pp. 2766--2769.

\bibitem[Stu10]{MR2647868}
D.~Stuart, \href{http://dx.doi.org/10.1063/1.3294085}{{\em Existence and
  {N}ewtonian limit of nonlinear bound states in the {E}instein-{D}irac
  system\/}}, J. Math. Phys. {\bf 51} (2010), pp. 032501, 13.

\bibitem[SV86]{MR848095}
W.~A. Strauss and L.~V{\'a}zquez,
  \href{http://dx.doi.org/10.1103/PhysRevD.34.641}{{\em Stability under
  dilations of nonlinear spinor fields\/}}, Phys. Rev. D (3) {\bf 34} (1986),
  pp. 641--643.

\bibitem[Thi58]{MR0091788}
W.~E. Thirring, \href{http://dx.doi.org/10.1016/0003-4916(58)90015-0}{{\em A
  soluble relativistic field theory\/}}, Ann. Physics {\bf 3} (1958), pp.
  91--112.

\bibitem[VK73]{VaKo}
N.~G. Vakhitov and A.~A. Kolokolov,
  \href{http://dx.doi.org/10.1007/BF01031343}{{\em Stationary solutions of the
  wave equation in the medium with nonlinearity saturation\/}}, Radiophys.
  Quantum Electron. {\bf 16} (1973), pp. 783--789.

\bibitem[Wak66]{wakano-1966}
M.~Wakano, \href{http://ptp.ipap.jp/link?PTP/35/1117/}{{\em Intensely localized
  solutions of the classical {D}irac-{M}axwell field equations\/}}, Progr.
  Theoret. Phys. {\bf 35} (1966), pp. 1117--1141.

\bibitem[Wei85]{MR783974}
M.~I. Weinstein, \href{http://dx.doi.org/10.1137/0516034}{{\em Modulational
  stability of ground states of nonlinear {S}chr\"odinger equations\/}}, SIAM
  J. Math. Anal. {\bf 16} (1985), pp. 472--491.

\end{thebibliography}
\bibliographystyle{sima-doi}

\end{document}